\documentclass[a4paper,11pt]{article}

\usepackage[english]{babel}
\usepackage[english]{nomencl}
\usepackage[utf8]{inputenc}
\usepackage[T1]{fontenc}

\usepackage{amsmath}
\usepackage{amssymb}
\usepackage{amsthm}
\usepackage[usenames,dvipsnames]{xcolor}
\usepackage[ruled,vlined]{algorithm2e}
\usepackage[hang,format=plain,labelfont=bf,up,textfont=small,up]{caption}
\usepackage[position=top]{subfig}
\usepackage{mdwlist}
\usepackage[]{hyperref}

\usepackage{pgf}
\usepackage{tikz,fullpage}
\usetikzlibrary{plothandlers,arrows,topaths,trees,shapes,fit,decorations.pathreplacing,decorations.text,decorations.pathmorphing,decorations.shapes,automata,positioning,chains}%,petri,automata,decorations,backgrounds,shadows

\newtheorem{Thm}{Theorem}%[section]

 \newtheorem{Lem}{Lemma}%[section]
 \newtheorem{Def}{Definition}%[section]

\newcommand{\nni}{{\rm nni}}
\renewcommand{\L}{\ensuremath{\mathcal{L}}}
\newcommand{\I}{\ensuremath{\mathcal{I}}}

\renewcommand{\O}{\ensuremath{\mathcal{O}}}
\renewcommand{\H}{\ensuremath{\mathcal{H}}}
\newcommand{\dnni}{\ensuremath{d_{nni}}}
\newcommand{\NNI}{\ensuremath{\mathcal{N}}}
\newcommand{\wt}{{\sf wt}}
\newcommand{\parent}{{\sf parent}}
\newcommand{\pnext}{{\sf next}}
\renewcommand{\path}{{\sf path}}
\newcommand{\ppath}{{\sf path}}
\newcommand{\head}{{\sf head}}
\newcommand{\sib}{{\sf sib}}
\newcommand{\leaf}{{\sf leaf}}
\newcommand{\merge}{{\sf merge}}
\newcommand{\dist}{{\sf dist}}
\newcommand{\junc}{{\sf junc}}
\newcommand{\pend}{{\sf end}}
\newcommand{\length}{{\sf length}}
\newcommand{\rank}{{\sf rank}}
\newcommand{\opt}{{\sf opt}}
\newcommand{\pre}{{\sf pre}}
\newcommand{\post}{{\sf post}}
\newcommand{\LCA}{{\sf LCA}}

\title{Efficient Parallel Computation of Nearest Neighbor Interchange Distances}
\author{Mikael Gast\thanks{Dept. of Computer Science, University of Bonn.
    e-mail:{ \tt gast@cs.uni-bonn.de}} \and
	Mathias Hauptmann\thanks{Dept. of Computer Science, University of Bonn.
    e-mail:{ \tt hauptman@cs.uni-bonn.de}}}
\date{}

\begin{document}

\maketitle

\begin{abstract}
The nni-distance is a well-known distance measure for phylogenetic trees. We construct an efficient parallel approximation algorithm for the nni-distance in the CRCW-PRAM model running in $\O(\log n)$ time on $\O(n)$ processors. Given two phylogenetic trees $T_1$ and $T_2$ on the same set of taxa and with the same multi-set of edge-weights, the algorithm constructs a sequence of nni-operations of weight at most $O(\log n)\cdot \opt$, where \opt\ denotes the minimum weight of a sequence of nni-operations transforming $T_1$ into $T_2$. This algorithm is based on the sequential approximation algorithm for the nni-distance given by DasGupta et al. (2000).
Furthermore, we show that the problem of identifying so called good edge-pairs between two weighted phylogenies can be computed in $\O(\log n)$ time on $\O(n\log n)$ processors.
\end{abstract}

\section{Introduction}

\emph{Phylogenetic trees} (or \emph{phylogenies}) are a well-known model for the history of evolution of species. Such a tree represents the lineage of a set of todays species, or more generally a set of \emph{taxa}, which are located at the leaf-level of the tree. The set internal nodes and the topology describe the ancestral history and interconnections among the taxa. Usually phylogenetic trees have internal nodes of degree 3.
A \emph{weighted} phylogeny additionally imposes \emph{weights} on its edges, representing the evolutionary distance between two taxa or internal nodes. We call a phylogeny \emph{unrooted} or \emph{rooted}, for the latter case if a common eldest ancestor is known and is designated as the root of the tree.

Concerning the reconstruction of phylogenetic trees from a given set of genetic data, a number of different models and algorithms have been introduced over the past decades. Each method is based on a different objective criterion or distance function in the course of construction --- for example \emph{parsimony}, \emph{compatibility}, \emph{distance} and \emph{maximum likelihood}. Due to this fact, the resulting phylogenies may vary according the internal topology and leaf configuration, although they have been created over the same set of taxa. Hence it is a reasonable approach to compare different phylogenies for their similarities and discrepancies.  As well for this task many different measures have been proposed, including \emph{subtree transfer metrics} \cite{Allen2001a}, \emph{minimum agreement subtrees} \cite{Finden1985} et cetera.

In this paper we focus on a restricted subtree transfer measure to compare phylogenetic trees, namely, the \emph{nearest neighbor interchange distance} (\nni), which was introduced by D.F. Robinson in \cite{Robinson1971b}. A \emph{\nni-operation} swaps two subtrees, which are both adjacent to the same edge $e$ in the tree. See Figure \ref{fig_nni} for an illustration of the \nni-operation.
The \emph{\nni-distance} between two trees is the minimum number of \nni-operations required to transform one tree into the other.

\begin{figure}[htb]
 \centering
 \tikzstyle{every node}=[isosceles triangle,draw]
\tikzstyle{leaf}=[circle,fill=black,inner sep=2pt,minimum size=2pt]

\subfloat[possible \nni-operations]{
\begin{tikzpicture}[yscale=0.6, xscale=0.9,label distance=-5pt]%out=90,in=90,relative]
\node (a) at (0,0) {A};
\node (b) at (0,-2) {B};
\node [shape border rotate = 180] (c) at (6,0) {C};
\node [shape border rotate = 180] (d) at (6,-2) {D};
\node [leaf,label=below:$u$] (e) at (2,-1) {};
\node [leaf,label=below:$v$] (f) at (4,-1) {};
\path (e)  edge (a.east)
	   edge (b.east);
\draw (e) -- (f) node[midway,sloped,above,draw=none] {$e$};

\path (f)  edge (c.west)
	   edge (d.west);
\path[red,<->] (b) edge[dashed,bend right] (d);
\path[blue,<->] (b) edge[dashed,bend right] (c);
\end{tikzpicture}
}

\subfloat[swap subtrees B and C]{
\begin{tikzpicture}[yscale=0.6, xscale=0.9,label distance=-5pt]
\node (a) at (0,0) {A};
\node (b) at (0,-2) [blue]{C};
\node [shape border rotate = 180] (c) at (6,0) [blue]{B};
\node [shape border rotate = 180] (d) at (6,-2) {D};
\node [leaf,label=below:$u$] (e) at (2,-1) {};
\node [leaf,label=below:$v$] (f) at (4,-1) {};
\path (e)  edge (a.east)
	   edge (b.east);
\draw (e) -- (f) node[midway,sloped,above,draw=none] {$e$};

\path (f)  edge (c.west)
	   edge (d.west);
\path[blue] (b) edge[dotted,bend right] (c);
\end{tikzpicture}
}
\subfloat[swap subtrees B and D]{
\begin{tikzpicture}[yscale=0.6, xscale=0.9,label distance=-5pt]
\node (a) at (0,0) {A};
\node (b) at (0,-2) [red]{D};
\node [shape border rotate = 180] (c) at (6,0) {C};
\node [shape border rotate = 180] (d) at (6,-2) [red]{B};
\node [leaf,label=below:$u$] (e) at (2,-1) {};
\node [leaf,label=below:$v$] (f) at (4,-1) {};
\path (e)  edge (a.east)
	   edge (b.east);
\draw (e) -- (f) node[midway,sloped,above,draw=none] {$e$};

\path (f)  edge (c.west)
	   edge (d.west);
\path[red] (b) edge[dotted,bend right] (d);
\end{tikzpicture}
}
 \caption{The possible non-redundant \nni-operations relative to an internal edge $e=(u,v)$. Each triangle A,B,C,D represents a subtree of the tree. The uniform cost of this operation is the weight $\wt(e)$ of edge $e$.}%$\wt(e)$.}
 \label{fig_nni}
\end{figure}
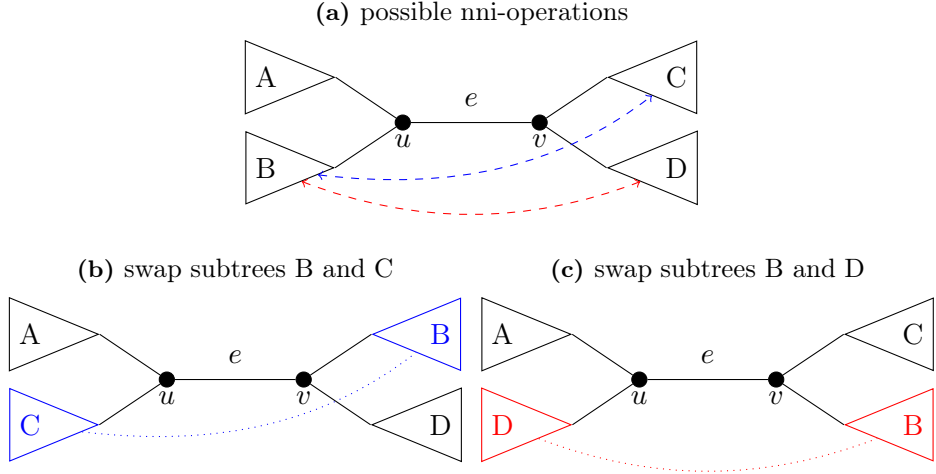

\subsection{Previous Results}

Although the \nni-distance has a simple definition in terms of a transformation of subtrees, the efficient and fast computation turned out to be surprisingly challenging.

For more than a decade, since its introduction in 1971 by Robinson \cite{Robinson1971b}, no efficient algorithm for computing the \nni-distance was known for practical (large) instances of phylogenetic trees. Day and Brown \cite{Day1985b} were the first to present an efficient approximation algorithm for unweighted instances. The algorithm runs in $\O(n\log n)$ time for unrooted and $\O(n^2\log n)$ time for rooted instances. %But it remained unclear whether the computational problem was \emph{NP-hard} or not.

Li, Tromp and Zhang \cite{Li1996} gave logarithmic lower and upper bounds on the maximum \nni-distance between arbitrary 3-regular trees. %and $\Delta(G)$ (the collection graph of 3-regular trees, with edges denoting that two trees are one $\nni$ apart). 
Furthermore, they gave an outline of a polynomial time approximation algorithm for unweighted instances with approximation ratio $\log n + \O(1)$.

DasGupta, He, Jiang, Li, Tromp and Zhang \cite{DasGupta2000} proved the NP-completeness of computing the \nni-distance on weighted and unweighted instances, and on trees with unlabeled (or non-uniformly labeled) leaves. They gave an approximation algorithm with running time $\O(n^2)$ and approximation ratio $4\log n +4$ for weighted instances. Furthermore, they observed that the \nni-distance is identical to the \emph{linear-cost subtree-transfer} distance on unweighted phylogenies \cite{DasGupta1999c} and gave an outline of an exact algorithm for distance-restricted instances with running time $\O(n^2 \log n + n \cdot 2^{11d})$.

\subsection{Our Work}

In this paper, we present an efficient \emph{parallel approximation algorithm} for the nni-distance on weighted phylogenies. This algorithm runs on a CRCW-PRAM in time $\O(\log n)$ with $\O(n\log n)$ processors and yields an approximation ratio of $\O(\log n)$. It is based on the sequential approximation algorithm by DasGupta et. al. \cite{DasGupta2000} with running time $\O(n^2)$ and approximation ratio $4(1 + \log n)$.
Especially, we obtain a CRCW-PRAM algorithm with time $\O(\log n)$ and $\O(n)$ processors for the case when no \emph{good edge-pairs} exist.
% To achieve this, we give efficient parallel algorithms for each step of DasGuptas algorithm and show that each step is in \emph{Nick's Class} (\emph{NC} for short) in terms of parallel computation complexity, i.e. can be computed in polylogarithmic time allocating a polynomial number of processors.

The  paper is organized as follows.
In Section \ref{sec:prelim} we give formal definitions of phylogenies and the \nni-distance. In Section \ref{subsec:seq_algo}, we describe the sequential approximation algorithm of DasGupta et. al. \cite{DasGupta2000}. In Section \ref{sec:par_nni} we present our new parallel approximation algorithm which consists of efficient parallel algorithms for linearizing trees (Section \ref{subsec:linear_tree}), sorting edge-permutations on linear trees (Section \ref{subsec:sorting_edges}) and sorting leaf-permutations on binary balanced trees (Section \ref{subsec:sorting_leaves}). Finally, in Section \ref{subsec:good_edge_pairs}, we present an efficient parallel algorithm to identify \emph{good edge-pairs} between two phylogenetic trees, in order to be able split up large instances and distribute the computational task already in a pre-computational step.
\section{Preliminaries}\label{sec:prelim}

We will make use of the following notation. Let $T=(V,E)$ be an undirected or directed tree, then $\L_T\subseteq V$ denotes the set of \emph{leaves} of $T$ and $\I_T\subseteq V$ the set of \emph{internal vertices} of $T$.

The most important primitives in phylogenetic analysis are taxa and phylogenies.

\begin{Def}
Given a finite set of \emph{taxa} $S=\{s_1,\ldots , s_n\}$, a \emph{phylogeny} for $S$ is a triplet 
$T=(V,E,\lambda)$ where $(V,E)$ is an undirected tree, $\lambda :\L_T\to S$ is a bijection and such that every internal node of $T$ has degree $3$.
A \emph{rooted phylogeny} for $S$ is a tuple $T=(V,E,\lambda,r)$ such that $(V,E,\lambda)$ is a phylogeny
and $r\in V$ is the root of $T$.
A \emph{weighted phylogeny} for $S$ is a tuple $T=(V,E,\lambda,\wt)$ such that $(V,E,\lambda)$ is a phylogeny and 
$\wt:E\to \mathbb{R}^+$ is a weight function on the set of edges of $T$. 
A \emph{rooted weighted phylogeny} is a tuple $T=(V,E,\lambda,\wt,r)$ such that $(V,E,\lambda,r)$ is a rooted phylogeny and 
$\wt:E\to \mathbb{R}^+$ is an edge-weight function.
\end{Def}

The \nni-distance is the minimum number of nearest neighbor interchanges (\nni) needed in order to transform one tree into another \cite{Robinson1979}:

\begin{Def}
Let $T$ be a phylogeny (possibly rooted and/or weighted) and let $e_1,e_2,e_3$ be three edges of $T$ that build a path of length three in $T$ (in this order). 
The associated \emph{\nni-operation}, denoted as a triplet $(e_1,e_2,e_3)$, transforms the tree $T$ into a new tree $T'$ by swapping the two subtrees below the edges $e_1$ and $e_3$ as shown in the Figure \ref{fig_nni_mini}. In this configuration we call the center edge $e_2$ the \emph{operating edge}. In case of weighted phylogenies the \emph{cost} of this \nni-operation is defined as $\wt(e_2)$.
\end{Def}

\begin{figure}[htb]
 \centering
 \tikzstyle{every node}=[isosceles triangle,draw]
\tikzstyle{leaf}=[circle,fill=black,inner sep=2pt,minimum size=2pt]

\begin{tikzpicture}[yscale=0.6, xscale=0.9,label distance=-5pt]%out=90,in=90,relative]
\node [inner sep=7pt] (a) at (0,0) {};
\node [inner sep=7pt,label=below right:$A$] (b) at (0,-2) {};
\node [shape border rotate = 180,inner sep=7pt,label=above left:$B$] (c) at (6,0) {};
\node [shape border rotate = 180,inner sep=7pt] (d) at (6,-2) {};
\node [leaf,label=below:$u$] (e) at (2,-1) {};
\node [leaf,label=below:$v$] (f) at (4,-1) {};
\path (e)  edge (a.east)
	   edge node[midway,sloped,above,draw=none] {$e_1$} (b.east);
\draw (e) -- (f) node[midway,sloped,above,draw=none] {$e_2$};

\path (f)  edge node[midway,sloped,above,draw=none] {$e_3$} (c.west)
	   edge (d.west);
%\path[red,<->] (b) edge[dashed,bend right] (d);
%\path[blue,<->] (b) edge[dashed,bend right] node[rectangle,near start,yshift=-7pt,sloped,draw=none] {$nni(e_1,e_2,e_3)$} (c);

\begin{scope}[xshift=10cm]
\node [inner sep=7pt] (A) at (0,0) {};
\node [blue,inner sep=7pt,label={[blue]below right:$B$}] (B) at (0,-2) {};
\node [blue,shape border rotate = 180,inner sep=7pt,label={[blue]above left:$A$}] (C) at (6,0) {};
\node [shape border rotate = 180,inner sep=7pt] (D) at (6,-2) {};
\node [leaf,label=below:$u$] (e) at (2,-1) {};
\node [leaf,label=below:$v$] (f) at (4,-1) {};
\path (e)  edge (A.east)
	   edge [blue] node[blue,midway,sloped,above,draw=none] {$e_3$} (B.east);
\draw (e) -- (f) node[midway,sloped,above,draw=none] {$e_2$};

\path (f)  edge [blue] node[blue,midway,sloped,above,draw=none] {$e_1$} (C.west)
	   edge (D.west);
%\path[red,<->] (b) edge[dashed,bend right] (d);
%\path[blue,<->] (b) edge[dashed,bend right] node[rectangle,near start,yshift=-7pt,sloped,draw=none] {$nni(e_1,e_2,e_3)$} (c);
\end{scope}

\path [blue] (c.south east) edge [bend left,shorten >=3mm,shorten <=3mm,decorate,
decoration={snake,amplitude=.4mm,segment length=3mm,
pre length=3mm,post length=3mm},-to] node [draw=none,below=-3mm,midway,sloped] {$nni(e_1,e_2,e_3)$} (A.south west);

\end{tikzpicture}
 \caption{The \nni-operation on $T$ of the subtrees $A$ and $B$ defined by the triplet $(e_1,e_2,e_3)$.}
 \label{fig_nni_mini}
\end{figure}
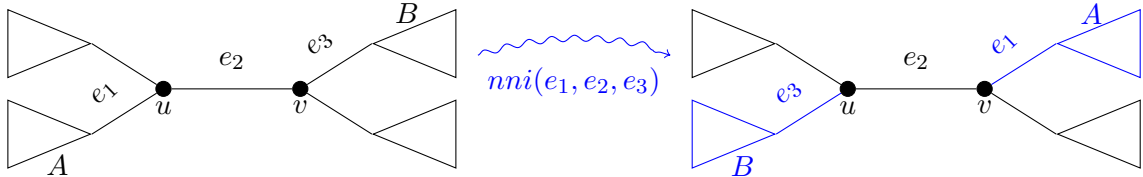

The associated genetic distance measure is the \emph{\nni-distance}:

\begin{Def}
Let $S$ be a set of taxa and let $T_1,T_2$ be phylogenies for $S$. 
The \emph{\nni-distance} $\dnni(T_1,T_2)$ of $T_1,T_2$ is the \emph{minimum length} of a sequence of \nni-operations that transforms $T_1$ into $T_2$ (and $\infty$ 
in case no such sequence exists). 
In case of weighted phylogenies $\dnni(T_1,T_2)$ is the \emph{minimum cost} of a sequence of \nni-operations 
that transforms $T_1$ into $T_2$.
\end{Def}

Given two weighted phylogenetic trees $T_i=(V_i,E_i,\lambda_i,\wt_i),\: i=1,2$ for the same set of taxa $S$, 
the following two conditions are necessary for the two trees to have a finite \nni-distance.

\begin{enumerate*}
 \item For each taxon $s\in S$,  let $e_i(s)\in E_{i}$ be the edge incident to the leaf with label $s$ in $T_i$ ($i=1,2$). Then $e_1(s)$ and $e_2(s)$ must have the same edge weight: $\wt_1(e_1(s))=\wt_2(e_2(s))$.
 \item $M_{1}=M_{2}$, where $M_{i}$ denotes the multiset of edge-weights of $T_i$.
\end{enumerate*}

In order to identify parts or subtrees of the tree that require a ``large'' or ``small'' amount of work to be transformed into their counterparts from the other tree, the notion of \emph{good edge-pairs} and \emph{bad edges} or \emph{non-shared edges} according to the set of leaf-labels and edge-weights is used in the literature (cf. \cite{Robinson1979,DasGupta2000}).

\begin{Def}{\bf (Good Edge-Pairs, Bad Edges)}\\
 Let $T_1$ and $T_2$ be two weighted phylogenies for the set of taxa $S$.
 Two internal edges $e_i \in E_{T_1}$ and $e_j \in E_{T_2}$ form a \emph{good edge-pair} if and only if the following conditions hold:
 \begin{enumerate*}
  \item $\wt_1(e_i) = \wt_2(e_j)$.
  \item Both edges induce the same partition of the multiset of edge-weights on $T_1$ and $T_2$.
  \item Both edges induce the same partition of the set of leaf-labels on $T_1$ and $T_2$.
 \end{enumerate*}
An edge $e_i \in E_{1}$ is called \emph{bad} if there does not exist any edge $e_j \in E_{2}$ 
such that $(e_i,e_j)$ forms a good edge-pair.
\label{def_gep}
\end{Def}
If $e_i$ and $e_j$ form a good edge pair, no \nni-move with operating edge $e_i$ is needed to transform $T_1$ into $T_2$.

\subsection{DasGupta's Sequential Approximation Algorithm}\label{subsec:seq_algo}
In this section we give an outline of DasGupta's approximation algorithm \cite{DasGupta2000} for the \nni-distance on weighted phylogenies on a set $S$ of $n$ taxa. 
For the ease of notation we assume that the phylogenies are rooted. Unless otherwise mentioned we will refer to these rooted and weighted phylogenies on $S$ as \emph{phylogenies} for short. Hence for the rest of this paper, a phylogeny is always a rooted and weighted phylogeny $T=(V,E,\lambda,r)$.

\begin{Thm}
 \emph{\cite{DasGupta2000}} Let $T_1$ and $T_2$ be two phylogenies. Then $\dnni(T_1,T_2)$ can be approximated within $\O(n^2)$ time and A.R. $4(1 + \log n)$.
\label{thm_main}
\end{Thm}

Given two phylogenies $T_1, T_2$, at first the multisets of edge-weights of internal edges of both, $T_1$ and $T_2$, are sorted in $\O(n \log n)$ time. In case these two multisets differ, $T_1$ and $T_2$ do not have a finite \nni-distance. Hence, from now on we assume that $\{w_1,w_2,\dots,w_{n-3}\}$ is the multiset of edge-weights of internal edges of both $T_1$ and $T_2$ and that
$w_1\leq w_2\leq \dots \leq w_{n-3}$ holds. Furthermore let $W:=\sum_{i=1}^{n-3}w_i$ be the sum of all edge weights of internal edges of $T_i, i\in\{1,2\}$.

\begin{Lem}
\emph{\cite{DasGupta2000}} If $\dnni(T_1,T_2)<\infty$ and $T_1$ and $T_2$ have no good edge pairs, then $\dnni(T_1,T_2) \geq W$.
\label{lem_nni_1}
\end{Lem}

DasGupta's algorithm makes use of two different trees associated to each of the given phylogenies $T_1,T_2$, which we call the \emph{auxiliary tree} and the \emph{linear tree}.

Let $T=(V,E,\lambda,\wt,r)$ be a phylogeny. An \emph{auxiliary tree} ${T'}=(V,E',\lambda,\wt',r)$ is a phylogeny on the same set of vertices $V$ and labeling of taxa $\lambda$ that has the following properties:
 \begin{itemize*}
  \item all leaves $l,l'\in\L_{{T'}}$ are of balanced height, $|depth_{{T'}}(l)-depth_{T'}(l')|=1$,
  \item the multisets of edge-weights in the trees $T$ and $T'$
 are the same, $M=M'$,
  \item the edge-weights of internal edges on every path from $r$ to a leaf in $T'$ are non-descending.
 \end{itemize*}

If the set $M$ of edge-weights is sorted such that $w_1\leq w_2\leq \dots \leq w_{n-3}$ holds, we achieve the auxiliary tree property by arranging the edge-weights in $M$ on an binary balanced tree such that, at level $i$, $w_{2^{i}-1+j}$ is the $j$-th edge-weight assigned to an edge from the left.
DasGupta's algorithm constructs auxiliary trees ${T'_i}=(V_i,E_i',\lambda_i,\wt_i',r_i)$, $i=1,2$, for $T_1$ and $T_2$.
Then both the original phylogenies $T_i$ and the associated auxiliary trees ${T'_i}$ are transformed into so called \emph{linear trees}:
For a given phylogeny $T=(V,E,\lambda,\wt,r)$, a \emph{linear tree} $L_T=(V,E'',\lambda,\wt'',r)$ of $T$ is a phylogeny with the same labeling $\lambda$ and such that every internal node is adjacent to at least one leaf (cf. Figure \ref{fig_ltree}).

\begin{figure}[htb]
 \centering
 \begin{tikzpicture}[vertex/.style={circle,draw},leaf/.style={fill=black,circle, inner sep=2pt}]

\foreach \a/\b/\c/\d in {leaf/A/1/1, leaf/B/2/1, leaf/C/3/1, leaf/E/5/1, leaf/F/6/1, leaf/a1/0/0, leaf/a2/1/0, leaf/b/2/0, leaf/c/3/0, leaf/e/5/0, leaf/f1/6/0, leaf/f2/7/0}
 \node[\a] (\b) at (\c,\d) {};

\foreach \a/\b/\c in {A/B/e_{1}, B/C/e_{2}, E/F/e_{n-3}}
 \draw (\a) -- (\b) node[midway,sloped,above,draw=none] {$\c$};

\foreach \a/\b in {A/a1, A/a2, B/b, C/c, E/e, F/f1, F/f2}
 \path (\a) edge (\b);

\path (C) edge[densely dotted] node[midway,yshift=-5mm,below] {$\dots$} (E);

\end{tikzpicture}
 \caption{The linear tree $L$ with internal edges $e_1,e_2,\dots,e_{n-3}$.}
 \label{fig_ltree}
\end{figure}
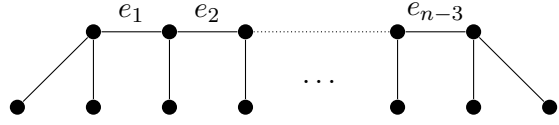

Then a variant of merge-sort is used to transform the order of internal edge of $L_{T_{1}}$ into the ordering of $L_{T_{2}}$. To transform the auxiliary tree ${T'_{1}}$ into ${T'_{2}}$ it remains to sort the order of leaves to complete the transformation from $T_1$ into $T_2$. Algorithm \ref{DSA} gives a pseudo-code description of DasGupta's algorithm.

\begin{algorithm}[!htb]
\KwIn{Rooted phylogenetic trees $T_1,T_2$.}
\KwOut{\nni-distance $\dnni(T_1,T_2)$ and a sequence \NNI\, of \nni-operations transforming $T_1$ into $T_2$.}
\Begin{
 \For{$i=1,2$}{
   \lnlset{bal}{1}Construct \emph{auxiliary trees} ${T'_i}$\;
   \tcc{generate \nni-sequence $\NNI_i$ to transform $T_i$ into ${T'_i}$}
   \lnlset{lin}{2}Generate sequence $(t_{i,1},\ldots,t_{i,j(i)})$ that transforms $T_i$ into a \emph{linear tree} $L_{T_i}$\;
   Generate sequence $(a_{i,1},\ldots,a_{i,k(i)})$ that transforms ${T'_i}$ into a \emph{linear tree} $L_{{T'_i}}$\;
   \lnlset{sort}{3}Generate \emph{merge-sort}-sequence $(s_{i,1},\ldots,s_{i,l(i)})$ that transforms $L_{T_i}$ into ${L_{{T'_i}}}$\;
   $\NNI_i := (t_{i,1},\ldots,t_{i,j(i)},\,s_{i,1},\ldots,s_{i,l(i)},a_{i,k(i)},\ldots,a_{i,1})$\;
   \tcc{note that sequence $(a_{i,1},\ldots,a_{i,k(i)})$ is reversed in order to allow back-transformation to ${T'_i}$}
 }
 \lnlset{leaf}{4}Generate sequence $(b_1,\ldots,b_m)$ to transform ${T'_1}$ into ${T'_2}$\;
 $\NNI := \NNI_1 \circ (b_1,\ldots,b_m) \circ \NNI'_2$\;
 \tcc{note that sequence $\NNI'_2$ is reversed for back-transformation to $T_2$}
}
\caption{\textsc{DasGupta's\_Sequential\_Algorithm}\label{DSA}}
\end{algorithm}

In case there exist good edge-pairs, these pairs yield a decomposition of $T_1,T_2$ into subtrees and Algorithm \ref{DSA} is applied to each pair of associated subtrees from $T_1$ and $T_2$. The parallel computation of good and bad edges will be treated in Section \ref{subsec:good_edge_pairs}. In the following, let us assume that there exists no good edge-pair between $T_1$ and $T_2$.
\section{Parallel Computation of the \nni-Distance}\label{sec:par_nni}

In this section we construct efficient parallel algorithms for the three steps of DasGupta's algorithm in the CRCW-PRAM-model. We start with a definition for the classification of internal nodes.

When $T$ is a 3-regular phylogeny (i.e each internal node has degree $3$ in $T$), 
the internal nodes of $T$ can be classified with respect to the number of adjacent leaves.

\begin{Def}
 Let $T=(V,E,\lambda,\wt)$ be a 3-regular phylogeny. Let $\L$ be the set of leaves in $T$. An internal node $v\in \I=(V \setminus \L)$ is called
 \begin{itemize*}
  \item an \emph{endnode} ($v \in V_{\pend}$), if it is adjacent to two leaves and one internal node,
  \item a \emph{pathnode} ($v \in V_{\path}$), if it is adjacent to one leaf and two internal nodes,
  \item a \emph{junction-node} ($v \in V_{\junc}$), if it is adjacent to three internal nodes in $T$.
 \end{itemize*}
\label{def_nodes}
\end{Def}
This notation will be used in the course of the linearization-step \ref{lin} of the sequential algorithm.

\subsection{Linearizing Trees}\label{subsec:linear_tree}

In the first algorithmic step, both $T_1, T_2$ and their associated auxiliary trees $T_1',T_2'$ are transformed into \emph{linear trees} $L_1, L_2, L_1', L_2'$ respectively (cf. Figure \ref{fig_ltree}). Let us first give an outline of our parallel linearization procedure, which consists of three phases:

\begin{description}
 \item{\emph{1. Activation-Phase:}} We proceed in a bottom-up manner at the boundary of the tree, i.e. at endnodes $v \in V_{\pend}$ defined above. At every endnode $v$ a process is started that builds the path to the next junction-node $u \in V_{\junc}$ and \emph{activates} $u$ to prepare the junction node for insertion of the path from $v$.
 
 If a junction-node $u$ is activated by more than one endnode in the activation phase, among the two paths meeting at $u$ we select the one of smaller weight for insertion. Let this path consist of $k$ internal edges $e_1,\dots,e_k$ where $e_1$ is incident to $u$.
 \item{\emph{2. Insertion-Phase:}} We generate the sequence of \nni-operations that is used for the insertion of the selected path at the junction-node $u$. This yields a sequence of \nni-operations of length $k$, the length of the path to be inserted. The internal edges $e_1,\dots,e_k$ are the operating edges of these \nni-moves.
 \item{\emph{3. Update-Phase:}} In the last phase the tree topology and the pointers inside the tree are updated.% after each Insertion-Phase.
\end{description}

These three phases are repeated until the trees $T_1, T_2, T_1', T_2'$ are transformed into linear trees $L_1, L_2, L_1', L_2'$, respectively.

\paragraph{Generating the Endnode-Paths for Insertion}

Algorithm \ref{TJD} computes for every node $v$ the distance $\dist(v)$, edge-list $\path(v)$, length $\length(v)$ and the head $\head(v)$ of the path to the next junction- or endnode $\pnext(v)$ heading towards root $r$. These values are computed efficiently in parallel via \emph{parallel pointer jumping} in $\O(\log n)$ time on $n$ processors.

\begin{algorithm}[!htb]
\KwIn{Phylogeny $T$ with root $r$ and pointer $\parent(v)$ for all $v$ in $T$ and sets of junction- and endnodes $V_{\junc}$ and $V_{\pend}$.}
\KwOut{For every node $v$ in $T$ the values $\dist(v)$, $\ppath(v)$, $\length(v)$, $\pnext(v)$ and $\head(v)$.}
\BlankLine
\Begin{
 \ForEach{$v \in V$ \bf parallel}{
  $\dist(v) := \wt(e_v)$\tcc*{initialize with parent edge $e_v=(v,\parent(v))$}
  $\ppath(v) := e_{v}$\;
  $\head(v) := v$\;
  $\length(v) := 1$\;
  $\pnext(v) := \parent(v)$\;
  \While{$\pnext(v) \notin V_{\junc} \cup V_{\pend}$}{
   $\dist(v) := \dist(v) + \dist(\pnext(v))$\;
   $\ppath(v) := \ppath(v) \circ \ppath(\pnext(v))$\;
   $\head(v) := \pnext(v)$\;
   $\length(v) := \length(v) + \length(\pnext(v))$\;
   $\pnext(v) := \pnext(\pnext(v))$\tcc*{Pointer-Jumping}
  }
 }
}
\caption{\textsc{Endnode\_Paths}\label{TJD}}
\end{algorithm}

\begin{figure}[!htb]
 \centering
 \subfloat[situation at junction-node $u$]{
\begin{tikzpicture}[vertex/.style={circle,draw},
	leaf/.style={fill=black,circle, inner sep=2pt},
	every pin/.style={gray,pin distance=10mm},
	every pin edge/.style={<-,gray,decorate,decoration={snake,segment length=6pt,pre length=4pt,amplitude=1pt}},
	every fit/.style={dashed,ellipse,rounded corners},
	node distance=10mm,
	on grid]

\node[leaf,label=above left:{$u$},pin={above right:{$\pnext(v_k)=\pnext(w)$}}] (u) at (0,0) {};
\node[leaf,above=of u, label=above right:$r$] (root) {};

\node[leaf,below left=of u,label=left:{$v_1$},pin={above left:{$\head(v_k)$}}] (u1) {};
\node[leaf,below right=of u,label=right:{$x$}] (u2) {};
\node[leaf,below=of u1] (v1) {};
\node[leaf,below=of u2, label=below right:$w$] (w) {};
\node[leaf,below=of v1,label=below left:$v_k$] (v) {};

\path[->] (u) edge[dotted] (root)
	(u1) edge node[label=above left:$e_1$] {} (u) %node[pin={above left:$path.last(v)$}] (u1_pin) {} (u)
	(v1) edge[dotted] (u1)
	(v) edge node[label=left:$e_k$] {} (v1)
	(u2) edge node[label=above right:$e_x$] {} (u)
	(w) edge[dotted] (u2);

%\node [draw=blue,fit=(v.west) (u1) (u1_pin),label={[blue]left:$path(v)$}] {};

%\draw[blue,fill,dashed,decorate,decoration={triangles,segment length=7pt}] plot[smooth cycle] coordinates{([xshift=-4pt]v.south west) ([xshift=-4pt]u1.north west) ([xshift=-4pt]u.north)};
\draw[decorate,decoration={triangles,segment length=7pt},red,fill] ([xshift=4pt] v.south east) .. controls ([xshift=3pt]u1) .. ([xshift=4pt] u.south);

\end{tikzpicture}
}
%\hspace*{1cm}
\subfloat[after insertion of $\ppath(v_k)$]{
\begin{tikzpicture}[vertex/.style={circle,draw},
	leaf/.style={fill=black,circle, inner sep=2pt},
	every pin/.style={gray,pin distance=15mm},
	every pin edge/.style={<-,solid,gray,decorate,decoration={snake,segment length=6pt,pre length=4pt,amplitude=1pt}},
	every fit/.style={dashed,ellipse,rounded corners},
	node distance=10mm,
	on grid]

\node[leaf,label=above right:{$u$}] (u) at (0,0) {};
\node[leaf,left=of u, label=above:$r$] (root) {};

\node[leaf,below right=of u,label=below left:{$v_1$}] (u1) {};
\node[leaf,right=of u1,label=above right:$v_k$] (v) {};
\node[leaf,below right=of v,label=below left:{$x$}] (u2) {};
\node[leaf,right=of u2, label=above:$w$] (w) {};

\path[->] (u) edge[dotted] node[pin={below:{$\pnext(w)$}}] (r_pin) {} (root)
	(u1) edge (u)
	(v) edge[dotted] (u1)
	(u2) edge (v)
	(w) edge[dotted] (u2);

\end{tikzpicture}
}
 \caption{Insertion of $\path(v_k)$ from endnode $v_k$ adjoining junction-node $u=\pnext(v_k)$.}
 \label{fig_linearize}
\end{figure}
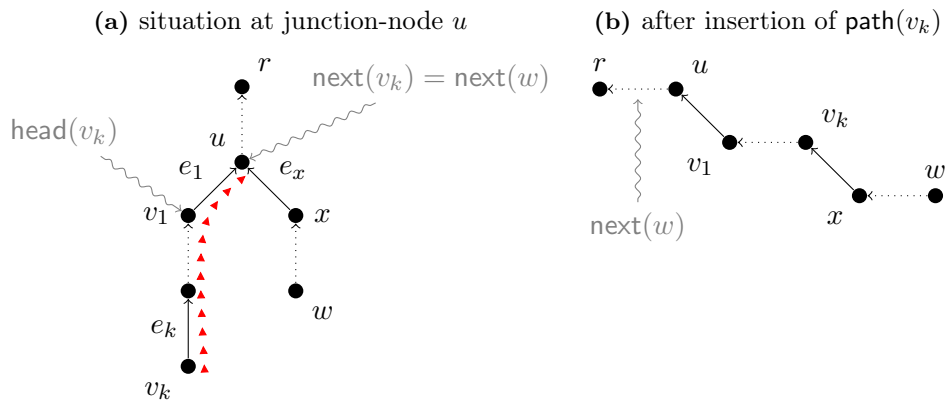

\paragraph{Parallel Linearization of Trees}

We are now ready to formulate Algorithm \ref{PLT} for the linearization of a tree $T$. Figure \ref{fig_linearize} illustrates the notation used in Algorithm \ref{TJD} and \ref{PLT}, and shows the result of an insertion-process.

\begin{algorithm}[!htb]
\KwIn{A phylogeny $T$ with root $r$.}
\KwOut{A list \NNI of \nni-operations which transforms $T$ into $L_T$.}
\BlankLine
 \While{$\exists  u \in V_{\junc}$}{
  $\textsc{Endnode\_Paths}(T)$\tcc*{re-generate paths and pointers}
  \ForEach{$v_k \in V_{\pend}$ \bf parallel}{
   $u := \pnext(v_k)$\;
   $a(u) := v_k$\tcc*{activate $u$ from $v_k$, $k=\length(v_k)$}
  }
  \ForEach{active $u \in V_{\junc}$ \bf parallel}{
   $x := \big(\sib(u)\neq\head(a(u))\big)$\;
   \ForEach{$1\leq i \leq k$ \bf parallel}{
    $\NNI_u[i] := \big(\big(\leaf(v_i),v_i\big), e_i,\, e_x\big)$\tcc*{generate \nni-triplets for every operating edge $e_i$ on the path to $v_k$}
   }
   $\NNI := \NNI \circ \NNI_u$\tcc*{concatenate list of \nni's}
   $\parent(x) := v_k$\tcc*{insertion of the path at $x$}
   $\wt((x,v_k)) := \wt((x,u))$\;
   $V_{\junc} := V_{\junc} \setminus \{u\}$\tcc*{deletion of $u$ from the set of junction-nodes}
  }
 }
\caption{\textsc{Parallel\_Linear\_Tree}\label{PLT}}
\end{algorithm}

\begin{Lem}\label{lem:linear_tree}
 Algorithm \ref{PLT} transforms a given phylogeny $T$ into a linear tree $L_T$ in $\O(\log n)$ time on $n$ processors.
\end{Lem}
\begin{proof}
 In every iteration endnode-paths are newly generated in time $\O(\log n)$ on $n$ processors. Then junction-nodes are activated and paths are inserted in parallel for every active junction-node, i.e. for every active endnode in constant time using $n$ processors.

 Now let $|V_{\pend}|=l_0$ be the initial number of endnodes in $T$ in iteration $0$ of the linearization-step. Now every endnode $v\in V_{\pend}$ tries to activate the next junction-node $\pnext(v)$ towards the root of $T$. This will be successful for at least every second endnode, since one junction-node is shared by at most two endnodes. Therefore at least $\frac{l_0}{2}$ insertions of an endnode-path $\path(v)$ is carried out at $\pnext(v)$ in each iteration and the number of end- and junction-nodes is reduced by at least $\frac{l_i}{2}$ in iteration $i$. Thus the number of iterations is bounded by $\lceil \log n \rceil$.
\end{proof}

\subsection{Sorting Edge-Permutations on Linear Trees}\label{subsec:sorting_edges}

This phase refers to step \ref{sort} of the sequential algorithm.
We are starting with two linear trees $L_1$ and $L_1'$ associated to the original tree $T_1$ and the balanced tree $T_1'$ with presorted edges. Now the sequence of nni-operations will be generated that transforms the sequence $e_1',e_2',\dots,e_{n-3}'$ of internal edges in $L_1$ into the linearized sorted sequence, say $e_1'',e_2'',\dots,e_{n-3}''$, of $L_1'$. 

The general approach of the sequential algorithm is first to transform adjacent edge-pairs by \nni-moves, such that afterwards the whole sequence is pairwise alternating from ascending to descending according to the sorting order of $e_1'',e_2'',\dots,e_{n-3}''$ (the ascending and descending subsequences of edges will be called \emph{blocks}).
Then, starting from the middle, we merge and pull out adjacent blocks via \nni-operations, finally resulting in a linear tree of blocks of doubled size, again alternating. At \emph{k}-th stage, we begin with $\frac{n}{2^{k}}$ blocks of $2^{k}$ internal edges each, resulting in $\frac{n}{2\cdot2^{k}}$ blocks consisting of $2\cdot2^{k}$ edges. See Figure \ref{fig_tree_merge} for an illustration. The sorting algorithm terminates if the resulting sequence consists of only one block, containing all edges. %, we are done with merge-sorting the linear tree.

\begin{figure}[!htb]
 \centering
 \subfloat[initially unsorted tree $L^{1}$ with $|B_i|=1, i=8$]{
\begin{tikzpicture}
[scale=0.89,
label distance=-5pt,
every node/.style={fill=black,circle,inner sep=2pt}]

\node (v1) at (0,0) {};
\node (v2) at (1,0) {};
\node (v3) at (2,0) {};
\node (v4) at (3,0) {};
\node (v5) at (4,0) {};
\node (v6) at (5,0) {};
\node (v7) at (6,0) {};
\node (v8) at (7,0) {};
\node (v9) at (8,0) {};

\foreach \a/\b/\c in {v1/v2/B^1_1, v2/v3/\dots, v3/v4/ , v4/v5/, v5/v6/, v6/v7/, v7/v8/\dots, v8/v9/B^1_8}
 \path (\a) edge node[above,fill=none,draw=none]{$\c$} (\b);

\end{tikzpicture}
}
\hspace*{1ex}
\subfloat[$L^{2}$ with pairwise alternating edge-weights]{
\begin{tikzpicture}
[scale=0.89,
label distance=-5pt,
every node/.style={fill=black,circle,inner sep=2pt},
every edge/.style={draw,densely dashed,->}]

\node (v1) at (0,0) {};
\node (v2) at (1,1) {};
\node (v3) at (2,0) {};
\node (v4) at (3,1) {};
\node (v5) at (4,0) {};
\node (v6) at (5,1) {};
\node (v7) at (6,0) {};
\node (v8) at (7,1) {};
\node (v9) at (8,0) {};

\foreach \a/\b/\c in {v1/v2/B^2_1, v2/v3/\dots, v3/v4/ , v4/v5/, v5/v6/, v6/v7/, v7/v8/\dots, v8/v9/B^2_8}
 \path (\a) edge node[above,sloped,fill=none,draw=none]{$\c$} (\b);

\end{tikzpicture}
}

\subfloat[$L^{3}$ after first merging-stage: $|B_{i}|=2, i=4$]{
\begin{tikzpicture}
[scale=0.89,
%label distance=-5pt,
every node/.style={fill=black,circle,inner sep=2pt},
every edge/.style={draw,densely dashed,->}]

\node (v1) at (0,0) {};
\node[gray] (v2) at (1,0.5) {};
\node (v3) at (2,1) {};
\node[gray] (v4) at (3,0.5) {};
\node (v5) at (4,0) {};
\node[gray] (v6) at (5,0.5) {};
\node (v7) at (6,1) {};
\node[gray] (v8) at (7,0.5) {};
\node (v9) at (8,0) {};

\foreach \a/\b/\c in {v1/v2/, v2/v3/, v3/v4/\dots, v4/v5/, v5/v6/, v6/v7/\dots, v7/v8/, v8/v9/}
 \path (\a) edge node[above,sloped,fill=none,draw=none]{$\c$} (\b);

\draw[decorate,decoration={brace,raise=5pt}] (v1) -- node[above=5pt,sloped,fill=none,draw=none]{$B^3_1$} (v3);
\draw[decorate,decoration={brace,raise=5pt}] (v7) -- node[above=5pt,sloped,fill=none,draw=none]{$B^3_4$} (v9);

\end{tikzpicture}
}
\hspace*{2ex}
\subfloat[$L^{4}$ after second merging-stage: $|B_{i}|=4, i=2$]{
\begin{tikzpicture}
[scale=0.89,
label distance=-5pt,
every node/.style={fill=black,circle,inner sep=2pt},
every edge/.style={draw,densely dashed,->}]

\node (v1) at (0,0) {};
\node[gray] (v2) at (1,0.25) {};
\node[gray] (v3) at (2,0.5) {};
\node[gray] (v4) at (3,0.75) {};
\node (v5) at (4,1) {};
\node[gray] (v6) at (5,0.75) {};
\node[gray] (v7) at (6,0.5) {};
\node[gray] (v8) at (7,0.25) {};
\node (v9) at (8,0) {};

\foreach \a/\b in {v1/v2, v2/v3, v3/v4, v4/v5, v5/v6, v6/v7, v7/v8, v8/v9}
 \path (\a) edge (\b);

\draw[decorate,decoration={brace,raise=5pt}] (v1) -- node[above=5pt,sloped,fill=none,draw=none]{$B^4_1$} (v5);
\draw[decorate,decoration={brace,raise=5pt}] (v5) -- node[above=5pt,sloped,fill=none,draw=none]{$B^4_2$} (v9);

\end{tikzpicture}
}
 \caption{Sorting edges on a linear tree $L$ via merging and pulling out alternating sequences of edge-weights $B_i$. Note, that the length $|B_i|$ of the sorted sequences doubles in every merging-stage.}
 \label{fig_tree_merge}
\end{figure}
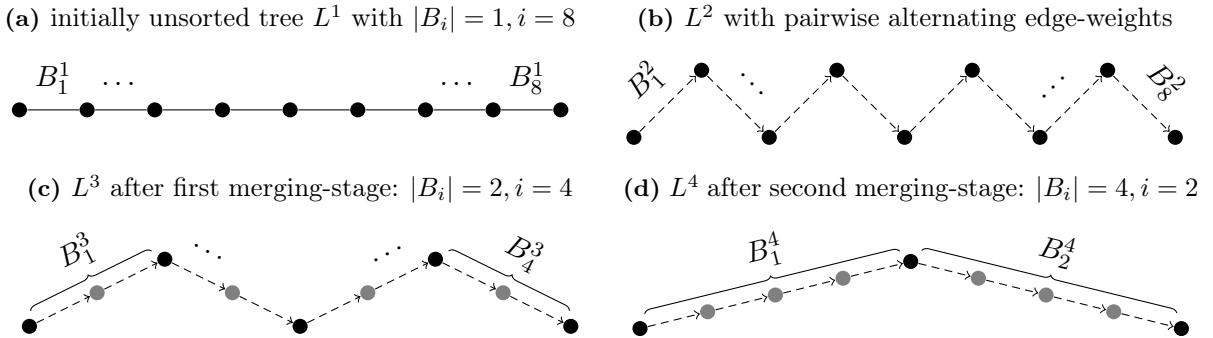

\paragraph{Parallel Tree Merging}

Now we describe an efficient parallel algorithm for sorting the edge permutations. We will not only consider the two adjacent blocks in the middle for comparing and merging, but all the $\frac{n}{2^{k}}$ block-pairs that will be adjacent in the course of stage $k$ in parallel. So we have to describe the pairing of blocks and edges inside blocks for each stage in order to allow for parallel computation.

At stage $k$ let $B_{1},B_{2},\dots,B_{\frac{n}{2^{k}}}$ be the blocks appearing in that order on the linear tree. We start pairing recursively from the middle, such that $B_{l}$ pairs with $B_{\frac{n}{2^{k}}-(l-1)}$ for $l\in\{1,\dots,\frac{n}{2\cdot2^{k}}\}$. Furthermore, let $e_{(l-1)2^{k}},e_{(l-1)2^{k}+1},\dots,e_{l2^{k}}$ be the edges of block $B_{l}$ at stage $k$.

To preserve simplicity, we illustrate the merging of edges of two blocks within a pair $(B^{k}_{x},B^{k}_{y})$, which is said to be a block-pair to get adjacent and to be merged at stage $k$ within the linear tree $L^{k}$.
Let $e_i^{k}$ denote the edge at position $i$ in $L^{k}$. The new position of this edge within $L^{k+1}$ is denoted by $e^{k+1}_{i+p}$, where $p$ is the \emph{rank} (regarding its edge-weight compared and ranked with the edge-weights of the opposite block) of $e^{k}$ in the opposite block of the merging-stage plus the number of equally ranked edges positioned before $e^{k}$ within the same block.

So if $e^{k}\in B^{k}_{x}$ is at position $i$ in $L^k$, we have $p =\rank(e^{k}_{i}|B^{k}_{y})+\big|\{e^{k}_{j}\in B^{k}_{x}\,|\,j<i, \rank(e^{k}_{j})=\rank(e^{k}_{i})\}\big|$ and the position changes from $e^{k}_{i}\rightsquigarrow e^{k+1}_{i+p}$ in $L^{k+1}$, as shown in Figure \ref{fig_block_merge}. We compute the ranking and positioning for all internal edges of the block-pair $(B^{k}_{x},B^{k}_{y})$ in parallel.

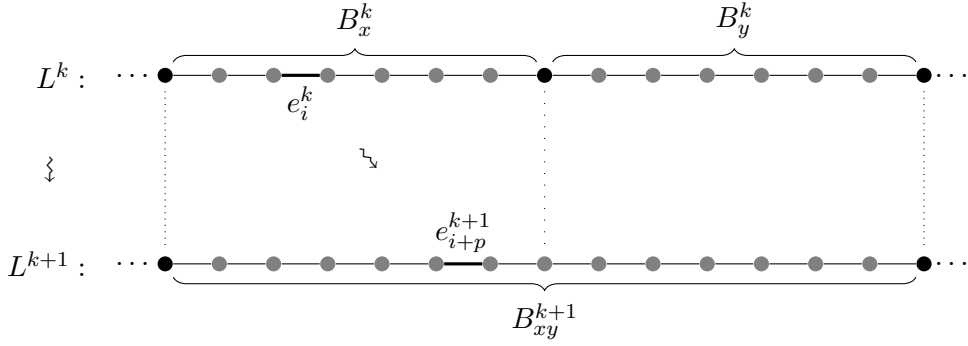
\begin{figure}[!htb]
 \centering
 \usetikzlibrary{chains,positioning,decorations.pathreplacing}

\begin{tikzpicture}[
vertex/.style={fill=gray,circle, inner sep=2pt},
leaf/.style={fill=black,circle, inner sep=2pt},
%every join/.style={black,->},
start chain=1,
start chain=2,
node distance=.5cm,
]

%chain one
\foreach \a in {leaf, vertex, vertex, vertex, vertex, vertex, vertex, leaf, vertex, vertex, vertex, vertex, vertex, vertex, leaf}
 \node[\a,on chain=1, join]{};

\node[] (Lk) [left=of 1-1,xshift=-.3cm] {$L^k:$};
\node[]  [left=of 1-1,xshift=.6cm] {$\dots$};
\node[]  [right=of 1-15,xshift=-.6cm] {$\dots$};

\draw[very thick] (1-3) -- node[midway,sloped,below,draw=none] (ek1) { $e^k_i$} (1-4);

\draw[thin,decorate,decoration={brace,raise=5pt,amplitude=5pt}] (1-1) -- (1-8) node[above=11pt,midway,draw=none]{$B^k_x$};
\draw[thin,decorate,decoration={brace,raise=5pt,amplitude=5pt}] (1-8) -- (1-15) node[above=10pt,midway,draw=none]{$B^k_{y}$};

%chain two
\begin{scope}[yshift=-2.5cm]
\foreach \a in {leaf, vertex, vertex, vertex, vertex, vertex, vertex, vertex, vertex, vertex, vertex, vertex, vertex, vertex, leaf}
 \node[\a,on chain=2, join]{};

\node[] (Lk2) [left=of 2-1,xshift=-.3cm] {$L^{k+1}:$};
\node[]  [left=of 2-1,xshift=.6cm] {$\dots$};
\node[]  [right=of 2-15,xshift=-.6cm] {$\dots$};

\draw[very thick] (2-6) -- node[midway,sloped,above,draw=none] (ek2) { $e^{k+1}_{i+p}$} (2-7);

\draw[thin,decorate,decoration={brace,raise=5pt,amplitude=5pt,mirror}] (2-1) -- (2-15) node[below=11pt,midway,draw=none]{$B^{k+1}_{xy}$};
\end{scope}

\draw[dotted,shorten >= 3pt,shorten <= 3pt] (1-1) -- (2-1);
\draw[loosely dotted,shorten >= 3pt,shorten <= 3pt] (1-8) -- (2-8);
\draw[dotted,shorten >= 3pt,shorten <= 3pt] (1-15) -- (2-15);

\draw[draw=none] (1-3) -- node[midway,sloped,above,draw=none] {$\rightsquigarrow$} (2-6);
\draw[draw=none] (1-1) -- node[midway,sloped,above,draw=none,yshift=-1.75cm] {$\rightsquigarrow$} (2-1);

\end{tikzpicture}
 \caption{Ranking edges within a block-pair $(B^{k}_{x},B^{k}_{y})$ on the linear tree $L^{k}$, resulting in $L^{k+1}$ with doubled block-size at the combined block $B^{k+1}_{xy}$.}
 \label{fig_block_merge}
\end{figure}

Furthermore, this sorting procedure is performed in parallel for all block-pairs which get adjacent in stage $k$ on $L^k$ with total number of $\frac{n}{2^k}\cdot 2^k = n$ processors running in $\O(1)$ time. In order to compute the sequence of \nni-operations, needed for the transformation of $L^{k} \rightsquigarrow L^{k+1}$ we look at both the block-pair $(B^{k}_{x},B^{k}_{y})$ and the combined block $B^{k+1}_{xy}$. The sequence of internal edges $e^{k+1}_{1},\dots,e^{k+1}_{2^{k+1}}$ of $B^{k+1}_{xy}$ yields the sequence of \emph{operating} edges. To complete the \nni-triplet, we find in parallel for every edge $e^{k+1}_i$ the next edge from the opposite block with respect to the situation on $L^{k}$ appearing in the sequence.
The \nni-triplet is generated via the edge of the 'outer' leaf of $e^{k+1}_{i}$ and the first $e^{k+1}_{j}$ from the opposite block.

The actual pairing situation if the \nni-operations would be performed sequentially on $L^{k}$ is shown in Figure \ref{fig_block_merge2}.

\begin{figure}[!htb]
 \centering
 \usetikzlibrary{chains,positioning,decorations.pathreplacing,shapes}

\begin{tikzpicture}[
vertex/.style={fill=gray,circle, inner sep=2pt},
leaf/.style={fill=black,circle, inner sep=2pt},
%every join/.style={black,->},
start chain=1,
start chain=2,
node distance=.5cm,
]

%chain one
\foreach \a in {vertex, vertex, leaf, vertex, vertex}
 \node[\a,on chain=1, join]{};

\node[leaf] (vl) at (-3.5,0) {};
\draw (vl) --  +(0,-.7);
\draw (vl) --  +(-.7,0);
\node[] [right=of vl,xshift=-.6cm] {$\dots$};

\node[leaf] (vr) at (5.5,0) {};
\draw (vr) --  +(0,-.7);
\draw (vr) --  +(.7,0);
\node[] [left=of vr,xshift=.6cm] {$\dots$};

%\node[] (Lj) [left=of 1-1,xshift=-.3cm] {$L^k:$};
\node[] [left=of 1-1,xshift=.6cm] {$\dots$};
\node[] [right=of 1-5,xshift=-.6cm] {$\dots$};

\node[vertex,below=of 1-3,yshift=-10] (1) {};
\node[vertex,below=of 1] (2) {};
\node[rotate=90,below=of 2,yshift=5,xshift=7] (3) {$\dots$};
\draw (1-3) --  (1);
\draw (1) --  (2);
\draw (1) --  +(.7,0);
\draw (2) --  +(.7,0);
%\node[isosceles triangle,draw,rotate=90,minimum size=1cm,below right=of 1-2,xshift=-1.5cm] (tri) {};

\path (1-2) -- node[midway,above,draw=none] (ek1) { $e^{k}_{i}$} (1-3);
\path (1-3) -- node[midway,above,draw=none] (ek1) { $e^{k}_{j}$} (1-4);
\draw (1-2) -- node[midway,left,draw=none] (ek1) { $l^{k}_{i}$} +(0,-.7);
\draw (1-4) -- node[midway,right,draw=none] (ek1) { $l^{k}_{j}$} +(0,-.7);
\draw (1-1) --  +(0,-.7);
\draw (1-5) --  +(0,-.7);

\draw[thin,decorate,decoration={brace,raise=20pt,amplitude=5pt}] (-4,0) -- (1-3) node[above=26pt,midway,draw=none]{$B^k_x$};
\draw[thin,decorate,decoration={brace,raise=20pt,amplitude=5pt}] (1-3) -- (6,0) node[above=25pt,midway,draw=none]{$B^k_{y}$};

\draw[thick,dotted,-stealth] (-4,-1) .. controls (.75,-1) .. (.75,-2.5);
\draw[thick,dotted,-stealth] (6,-1) .. controls (2.5,-1) .. (2.5,-2.5);

% \node[rounded rectangle,draw,dotted,text justified] () at (1,-3.5) {\textbullet \hspace{1ex}if $w(e^{k}_{i}) < w(e^{k}_{i'})$ and \tikz\draw[solid] (0,0) -- node[midway,left,yshift=5] {$B^{k}_{l}$} (.5,.5) -- node[midway,right,yshift=5] {$B^{k}_{l+1}$} (1,0);\newline next nni op $=(e^{k}_{i},e^{k}_{i'},b)$};
% \node[rounded rectangle,draw,dotted,text justified] () at (1,-5) {\textbullet \hspace{1ex}if $w(e^{k}_{i}) > w(e^{k}_{i'})$ and \tikz\draw[solid] (0,.5) -- node[midway,left] {$B^{k}_{l}$} (.5,0) -- node[midway,right] {$B^{k}_{l+1}$} (1,.5);\newline next nni op $=(a,e^{k}_{i},e^{k}_{i'})$};

\end{tikzpicture}
 \caption{Merging edges of a block-pair $(B^{k}_{x},B^{k}_{y})$ via \nni-operations.}
 \label{fig_block_merge2}
\end{figure}
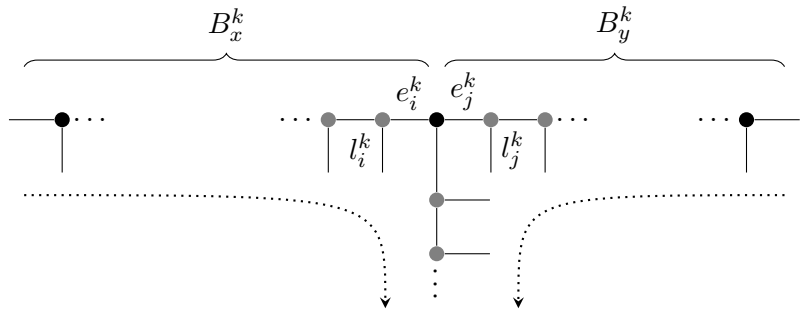

After pairing up every edge in the sequence we have:
\begin{itemize}
 \item if $\wt(e^{k}_{i}) < \wt(e^{k}_{j})$ and \tikz[baseline=0pt]\draw[solid] (0,0) -- node[midway,left=5pt] {$B^{k}_{x}$} (.5,.5) -- node[midway,right=5pt] {$B^{k}_{y}$} (1,0); holds, the next \nni-operation is $\nni(e^{k}_{i},e^{k}_{j},l^{k}_{j})$
 \item if $\wt(e^{k}_{i}) > \wt(e^{k}_{j})$ and \tikz[baseline=0pt]\draw[solid] (0,.5) -- node[midway,left=5pt] {$B^{k}_{x}$} (.5,0) -- node[midway,right=5pt] {$B^{k}_{y}$} (1,.5); holds, the next \nni-operation is $\nni(l^{k}_{i},e^{k}_{i},e^{k}_{j})$
\end{itemize}

We are now ready to state Algorithm \ref{TMS} to compute the sequence of \nni-operations used for merge-sorting a linear tree $L$.

\begin{algorithm}[!htb]
\KwIn{Linear tree $L$, permutation $e_1,e_2,\dots,e_{n-3}$ of internal edges of $L$.}
\KwOut{Sequence ${\NNI}$ of nni-operations that transforms $L$ into $L'$ with internal edges sorted.}
\BlankLine
 \For{$k=1$ \KwTo $\log n$}{
  \ForEach{$l\in\{1,\dots,\frac{n}{2\cdot2^{k}}\}$ \bf parallel}{
   $B_x := B_l$\;
   $B_y := B_{\frac{n}{2^{k}}-(l-1)}$\;
   $B_{xy} := \merge(B_{x}, B_{y})$\tcc*{Merging two blocks via \emph{ranking} edges}
   $L^k := B_{xy} \circ L^k$\tcc*{at the end of the foreach-Phase in the $k$-th iteration, $L^k=e^{k}_1,\dots,e^{k}_{n-3}$} 
  }
  \ForEach{$e^{k}_{i}\in L^{k}$ \bf parallel}{
  $e^{k}_{j} :=$ next edge from opposite block\;
  \If{$\wt(e^{k}_{i}) < \wt(e^{k}_{j})$ \bf and \tikz[baseline=0pt]\draw[solid] (0,0) -- node[midway,left=5pt] {$B^{k}_{x}$} (.5,.5) -- node[midway,right=5pt] {$B^{k}_{y}$} (1,0);}{
   $\nni(i) := (e^{k}_{i},e^{k}_{j},l^{k}_{j})$\tcc*{as illustrated in Figure \ref{fig_block_merge2}} 
  }
  \ElseIf{$\wt(e^{k}_{i}) > \wt(e^{k}_{j})$ \bf and \tikz[baseline=0pt]\draw[solid] (0,.5) -- node[midway,left=5pt] {$B^{k}_{x}$} (.5,0) -- node[midway,right=5pt] {$B^{k}_{y}$} (1,.5);}{
   $\nni(i) := (l^{k}_{i},e^{k}_{i},e^{k}_{j})$\;
  }
  $\NNI := \NNI \circ \nni(i)$\;
  }
 }
\caption{\textsc{Tree\_Merge\_Sort}\label{TMS}}
\end{algorithm}

We obtain the following Lemma:

\begin{Lem}\label{lem:edge_sort}
 The sorting of edge-permutations is performed in $\O(\log n)$ time on $n$ processors.
\end{Lem}
\begin{proof}
 In Algorithm \ref{TMS} the length of the sorted sub-sequences $|B_l|$ doubles with every merging-stage. Therefore at most
 $\log n$ complete merging-rounds are needed to yield a sorted sequence of length $n$. In the merging-steps of stage $k$, we have $\frac{n}{2^k}$ blocks of length $2^k$ which are compared and merged to blocks of doubled size using $\frac{n}{2^k}\cdot 2^k = n$ comparisons, i.e. allocating $n$ processors and yielding a running time of $\O(1)$. The subsequent generation of the \nni-triplets also uses $n$ processors for $\O(1)$ time per stage.% \dots to sort the whole sequence.
\end{proof}

\subsection{Sorting Leaf-Permutations on Balanced Binary Trees}\label{subsec:sorting_leaves}

This phase refers to step \ref{leaf} of the sequential algorithm.
We are given two binary balanced trees $T_1', T_2'$ that only differ in the ordering of leaves. The sequential algorithm generates a sequence $\NNI$ of \nni-operations which implement the cycles of the permutation of leaves transforming $T_1'$ into $T_2'$. We show how to generate this sequence efficiently in parallel.

Let $d$ be the \emph{depth} of $T_1'$ and $T_2'$. When $T_1'$ is transformed into $T_2'$ by use of the sequence $\NNI$, the corresponding intermediate trees might be unbalanced. More precisely, let $\pi:\{1,\dots,n\}\to\{1,\dots,n\}$ be the permutation transforming the order of leaves $l_1,\dots,l_n$ in $T_1'$ into $l_{\pi(1)},\dots,l_{\pi(n)}$ in $T_2'$. Let $\pi$ consist of cycles $C_{1},\dots,C_{k}$. Then $\NNI=\NNI_{1}\circ\dots\circ\NNI_{k}$ where $\NNI_{i}$ implements cycle $C_{i}$. Let $C_{i}=(c_{i,1},\dots,c_{i,q})$ be one cycle, then $\NNI_{i}=\NNI_{i,1}\circ\dots\circ\NNI_{i,q}$, where $\NNI_{i,j}$ is a sequence of \nni-operations which transports the leaf $l_{c_{i,j}}$ to its new position in $T_2'$ (cf. Figure \ref{fig_leaf_sort}).

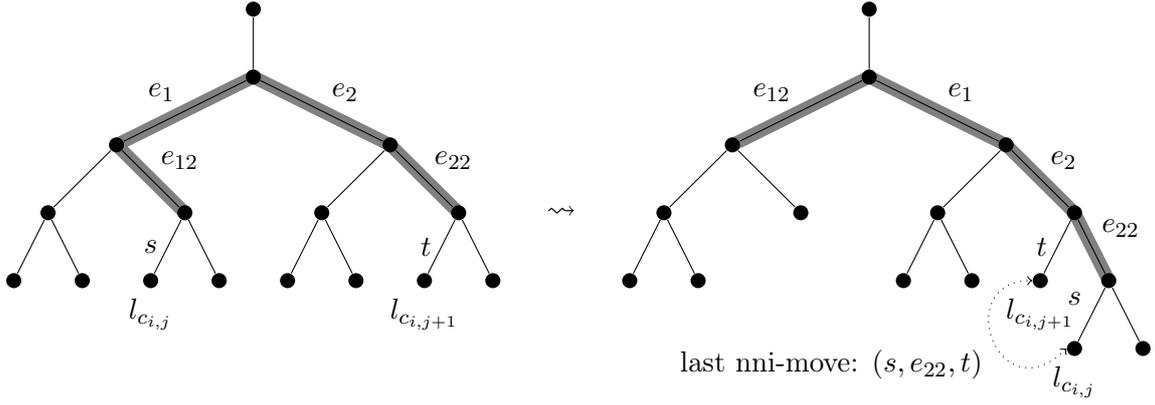
\begin{figure}[!htb]
 \centering
 \usetikzlibrary{patterns}

\pgfdeclarelayer{background layer}
\pgfdeclarelayer{foreground layer}
\pgfsetlayers{background layer,main,foreground layer}

\begin{tikzpicture}
[scale=0.9,
level distance=10mm,
every label/.style={fill=none,rectangle},
vertex/.style={fill=black,circle,inner sep=2pt},
level 0/.style={level distance=10mm,sibling distance=20mm},
level 1/.style={level distance=10mm,sibling distance=80mm},
level 2/.style={level distance=10mm,sibling distance=40mm},
level 3/.style={level distance=10mm,sibling distance=20mm},
level 4/.style={level distance=10mm,sibling distance=10mm}
]

\node[vertex] (r) {}
 child {node[vertex] (v) {}
  child {node[vertex] (v1) {}
   child {node[vertex] (v11) {}
    child {node[vertex] (v111) {}}
    child {node[vertex] (v112) {}}
   }
   child {node[vertex] (v12) {}
    child {node[vertex,label=below:$l_{c_{i,j}}$] (v121) {}
    edge from parent node [left] () {$s$}}
    child {node[vertex] (v122) {}}
   edge from parent node [above right] () {$e_{12}$} }
   edge from parent node [above left] () {$e_{1}$} }
  child {node[vertex] (v2) {}
   child {node[vertex] (v21) {} 
    child {node[vertex] (v211) {}}
    child {node[vertex] (v212) {}}
   }
   child {node[vertex] (v22) {}
    child {node[vertex,label=below:$l_{c_{i,j+1}}$] (v221) {}
    edge from parent node [left] () {$t$}}
    child {node[vertex] (v222) {}}
  edge from parent node [above right] () {$e_{22}$} }
  edge from parent node [above right] () {$e_{2}$} }
 };
\draw[draw=none] (v22)+(2,0) -- node [] () {$\rightsquigarrow$} +(1,0);

\begin{pgfonlayer}{background layer}
\draw[gray,line width=5pt,line cap=rect] (v22)  -- (v2) -- (v) -- (v1) -- (v12);
\end{pgfonlayer}

\begin{scope}[xshift=9cm]
\node[vertex] (r) {}
 child {node[vertex] (v) {}
  child {node[vertex] (v1) {}
   child {node[vertex] (v11) {}
    child {node[vertex] (v111) {}}
    child {node[vertex] (v112) {}}
   }
   child {node[vertex] (v12) {} }
   edge from parent node [above left] () {$e_{12}$} }
  child {node[vertex] (v2) {}
   child {node[vertex] (v21) {} 
    child {node[vertex] (v211) {}}
    child {node[vertex] (v212) {}}
   }
   child {node[vertex] (v22) {}
    child {node[vertex,label=below:$l_{c_{i,j+1}}$] (v221) {}
     edge from parent node [left] () {$t$}}
    child {node[vertex] (v222) {}
     child {node[vertex,label=below:$l_{c_{i,j}}$] (v2221) {}
      edge from parent node [above left] () {$s$}}
     child {node[vertex] (v2222) {}}
   edge from parent node [above right] () {$e_{22}$}}
  edge from parent node [above right] () {$e_{2}$} }
  edge from parent node [above right] () {$e_{1}$} }
 };

\begin{pgfonlayer}{background layer}
\draw[gray,line width=5pt,line cap=rect] (v222)  -- (v22) -- (v2) -- (v) -- (v1);
\draw[<->,dotted] (v2221.west) .. controls +(-1,-1) and +(-1.2,0) .. node [below left] {last nni-move: $(s,e_{22},t)$} (v221.west);
\end{pgfonlayer}
\end{scope}

\end{tikzpicture}
 \caption{Transportation of leaf $l_{c_{i,j}}$ at position $s$ to its target position $t$ with leaf $l_{c_{i,j+1}}$ attached.}
 \label{fig_leaf_sort}
\end{figure}

Let $T_{\H}$ denote the tree that results from applying sequence $\H$ of \nni-operations to the tree $T_1'$. For each prefix $\H$ of $\NNI$, the tree $T_{\H}$ has depth $d$ or $d+1$, hence the set of possible positions of edges in $T_{\H}$ is $P=\{(l,j)\,|\, 1\leq l \leq d+1, 1 \leq j \leq 2^l\}$. Then each of the trees $T_{\NNI_1 \circ\dots\circ \NNI_{j}}$ differs from $T_1'$ only w.r.t. the order (positions) of leaves, i.e. all the internal edges have the same position as in $T_1'$. Furthermore each $T_{\NNI_1 \circ\dots\circ \NNI_j \circ \NNI_{j+1,1} \circ\dots\circ \NNI_{j+1,h}}$ is one of the imbalanced trees $T_{s,t}$ of depth $d+1$ with $s,t\in P$ positions of depth $d-1$ and $d+1$ respectively (cf. Figure \ref{fig_leaf_sort2}).

\begin{figure}[!htb]
 \centering
 \usetikzlibrary{patterns,decorations.pathmorphing,decorations.pathreplacing}

\pgfdeclarelayer{background layer}
\pgfdeclarelayer{foreground layer}
\pgfsetlayers{background layer,main,foreground layer}
%,pin={above left:{$head(v)$}}
\begin{tikzpicture}
[scale=0.9,
level distance=10mm,
every pin/.style={gray,pin distance=15mm},
every pin edge/.style={<-,gray,decorate,decoration={snake,segment length=6pt,pre length=4pt,amplitude=1pt}},
every label/.style={fill=none,rectangle},
vertex/.style={fill=black,circle,inner sep=2pt},
level 0/.style={level distance=10mm,sibling distance=20mm},
level 1/.style={level distance=10mm,sibling distance=80mm},
level 2/.style={level distance=10mm,sibling distance=40mm},
level 3/.style={level distance=10mm,sibling distance=30mm},
level 4/.style={level distance=10mm,sibling distance=20mm},
level 5/.style={level distance=10mm,sibling distance=10mm}
]

\node[vertex] (r) {}
 child {node[vertex] (v) {}
  child {node[vertex] (x) {}
  child {node[vertex] (v1) {}
   child {node[vertex] (v11) {}
    child {node[vertex] (v111) {}}
    child {node[vertex] (v112) {}}
   }
   child {node[vertex] (v12) {} }
   edge from parent [draw=none] node [sloped] () {$\dots$} }}
  child {node[vertex] (y) {}
  child {node[vertex] (v2) {}
   child {node[vertex] (v21) {} 
    child {node[vertex] (v211) {}}
    child {node[vertex] (v212) {}}
   }
   child {node[vertex] (v22) {}
    child {node[vertex] (v221) {}}
    child {node[vertex] (v222) {}
     child {node[vertex] (v2221) {}}
     child {node[vertex] (v2222) {}}
    }
   }
  edge from parent[draw=none] node [sloped] () {$\dots$} }
  }
 };

\node () at (-4,0) {$T_{s,t}:$};

\node [label=right:$0$] (0) at (5,-.5) {};
\node [label=right:$1$] (1) at (5,-1.5) {};
\node [label=right:$d-1$] (2) at (5,-3.5) {};
\node [label=right:$d$] (3) at (5,-4.5) {};
\node [label=right:$d+1$] (4) at (5,-5.5) {};

\begin{pgfonlayer}{background layer}
\draw[gray,line width=5pt,line cap=rect] (v2221) -- node [left,black] () {$t$} (v222) (v1) -- node [above right,black] () {$s$} (v12);
\draw[loosely dotted] (0) -- +(-9.5,0);
\draw[loosely dotted] (1) -- +(-9.5,0);
\draw[loosely dotted] (2) -- +(-9.5,0);
\draw[loosely dotted] (3) -- +(-9.5,0);
\draw[loosely dotted] (4) -- +(-9.5,0);\draw[thin,decorate,decoration={brace,raise=50pt,amplitude=5pt}] (0.north) -- (4.south) node[right=55pt,midway,draw=none]{depth};
\end{pgfonlayer}

\end{tikzpicture}
 \caption{Unbalanced tree $T_{s,t}$.}
 \label{fig_leaf_sort2}
\end{figure}
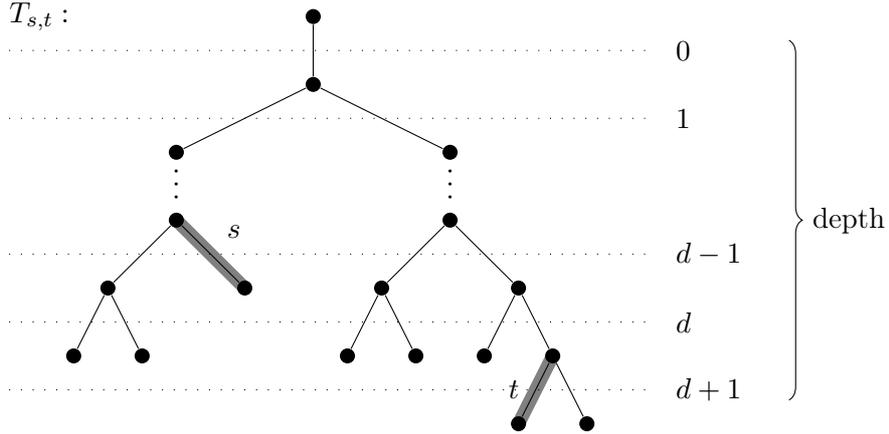

The positions of internal edges in $T_{s,t}$ only depend on $s$ and $t$:
if internal edge $e$ has position $(l,p)$ in $T_1'$, then its position in $T_{s,t}$ is one of $\big[(l,p), (l-q,\lfloor\frac{p}{2}\rfloor), (l+1,2p), (l+1,2p+1)\big]$ depending on if the edge $e$ is on the path from $s$ to $t$ and if it is on the ascending or descending part of this path.
Hence for each prefix $\H$ of $\NNI$ of the form $\H=\NNI_1 \circ\dots\circ \NNI_j \circ \NNI_{j+1,1} \circ\dots\circ \NNI_{j+1,h}$ the positions of edges $p_{\H}:E\to P$ in the tree $T_{\H}'$ which results from $T_1'$ by application of $\H$ can be computed efficiently in parallel.

\begin{Lem}\label{lem:leaf_sort}
 The sorting of leaf-permutations on two binary balanced trees can be done in time $\O(\log n)$ on $n$ processors. 
\end{Lem}
\begin{proof}
 Since the height of balanced binary trees is bounded by $\lceil\log n\rceil$, the positions of edges $p_{\H}:E\to P$ in $T'_{\H}$ for a prefix $\H$ of $\NNI$ can be efficiently computed in time $\O(\log n)$ on $n$ processors.
Thus it remains to describe how to compute the sequence $\NNI_{j+1,h+1}$ for a given $\H=\NNI_1 \circ\dots\circ \NNI_j \circ \NNI_{j+1,1} \circ\dots\circ \NNI_{j+1,h}$ and $p_{\H}$ as above:

Let $C_{j+1}=(c_{j+1,1},\dots,c_{j+1,h+1},c_{j+1,h+2},\dots)$ be the $(j+1)$-th cycle of $\pi$ and let $T_{\H}=T_{r,s}$ and $T_{\H\circ\NNI_{j+1,h+1}}=T_{r,t}$ with $r=(d-1,\lfloor\frac{c_{j+1,1}}{2}\rfloor)$. If $s=(d+1,x)$ and $t=(d+1,y)$ with $x=2^{d'}\cdot\alpha-j_{x}$ and $y=2^{d'}\cdot\alpha-j_{y}$ with $j_{x},j_{y}\in\{0,\dots,2^{d'}-1\}$ then the lowest common ancestor is at position $(d',\alpha)$ (see Figure \ref{fig_leaf_sort3}) and $\NNI_{j+1,h+1}$ is a sequence of $2(d-d')-1$ \nni-operations. Finally, $\NNI_{j+1,h+1}$ can be constructed in $\O(\log n)$ on a single processor since $d'<d\leq \log n$ and the positions of edges in the tree $T_{\H}$ are known at that point.
\end{proof}

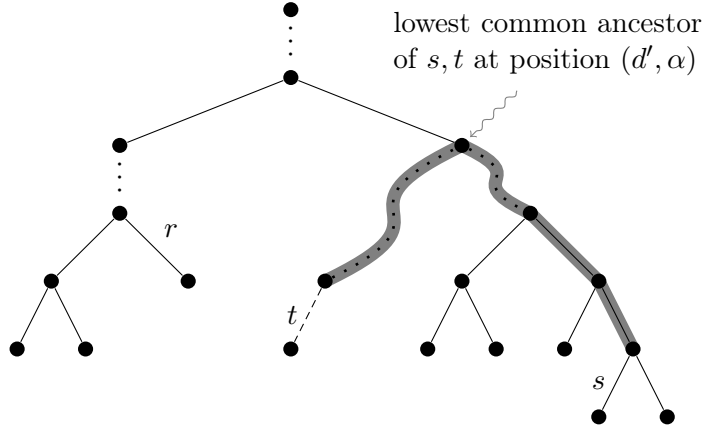
\begin{figure}[!htb]
 \centering
 \usetikzlibrary{patterns,decorations.pathmorphing,decorations.pathreplacing}

\pgfdeclarelayer{background layer}
\pgfdeclarelayer{foreground layer}
\pgfsetlayers{background layer,main,foreground layer}
\begin{tikzpicture}
[scale=0.9,
level distance=10mm,
every pin/.style={gray,pin distance=15mm},
every pin edge/.style={<-,gray,decorate,decoration={snake,segment length=6pt,pre length=4pt,amplitude=1pt}},
every label/.style={fill=none,rectangle},
vertex/.style={fill=black,circle,inner sep=2pt},
level 0/.style={level distance=10mm,sibling distance=20mm},
level 1/.style={level distance=10mm,sibling distance=80mm},
level 2/.style={level distance=10mm,sibling distance=50mm},
level 3/.style={level distance=10mm,sibling distance=20mm},
level 4/.style={level distance=10mm,sibling distance=20mm},
level 5/.style={level distance=10mm,sibling distance=10mm}
]

\node[vertex] (r) {}
 child {node[vertex] (v) {}
  child {node[vertex] (x) {}
  child {node[vertex] (v1) {}
   child {node[vertex] (v11) {}
    child {node[vertex] (v111) {}}
    child {node[vertex] (v112) {}}
   }
   child {node[vertex] (v12) {}
   edge from parent node [above right] () {$r$}  }
   edge from parent [draw=none] node [sloped] () {$\dots$} 
}}
  child {node[vertex] (y) {}
  child {node[draw=none] (z) {}
   edge from parent[draw=none] node [sloped] () {}
}
  child {node[vertex] (v2) {}
   child {node[vertex] (v21) {} 
    child {node[vertex] (v211) {}}
    child {node[vertex] (v212) {}}
   }
   child {node[vertex] (v22) {}
    child {node[vertex] (v221) {}}
    child {node[vertex] (v222) {}
     child {node[vertex] (v2221) {}
     edge from parent node [left] () {$s$}}
     child {node[vertex] (v2222) {}}
    }
   }
  edge from parent[draw=none] node [sloped] () {} }}
  edge from parent [draw=none] node [sloped] () {$\dots$} };

\node[vertex] (t1) at (0.5,-4) {};
\node[vertex] (t2) at (0,-5) {};
\draw[densely dashed] (t1) -- node [left] {$t$} (t2);

\node[text width=4.5cm] (text) at (4,-0.5) {lowest common ancestor of  $s,t$ at position $(d',\alpha)$};

\begin{pgfonlayer}{background layer}
\draw [preaction={draw,gray,line width=5pt,line cap=rect}]
            [draw,loosely dotted,line width=1pt] (t1) .. controls +(2,1) and +(-2,-1) ..  (y) .. controls +(1,-0.5) and +(-1,0.5).. (v2);
\draw[gray,line width=5pt,line cap=rect] (v2) -- (v22) -- (v222);
\draw[->,gray,shorten >=.5mm,decorate,decoration={snake,amplitude=1pt,segment length=6pt,pre=moveto,pre length=1mm,post length=1mm}] (text) -- (y);
\end{pgfonlayer}

\end{tikzpicture}
 \caption{Transportation-path between $s$ and $t$ via the lowest common ancestor in $T_{r,s}$.}
 \label{fig_leaf_sort3}
\end{figure}

This completes the last step of our parallel algorithm for approximating the \nni-distance between two weighted phylogenies and we get the following theorem as a corollary of Lemma \ref{lem:linear_tree}, \ref{lem:edge_sort} and \ref{lem:leaf_sort}.

\begin{Thm}
 The \nni-distance between two phylogenies $T_1$ and $T_2$ and the sequence of \nni-operations can be approximated within approximation ratio $\O(\log n)$ in $\O(\log n)$ time on $n$ processors.
\end{Thm}

In the last section, we present a parallel algorithm to compute \emph{good edge-pairs} in order to be able to split up large problem instances in a pre-processing step and to identify edges, for which no \nni-operation is needed in order to transform the trees into each other.

\subsection{Detecting Good Edge-Pairs}\label{subsec:good_edge_pairs}

Our aim is to identify \emph{good edge-pairs} $(e_x,e_y)$\footnote{not necessarily having $x=y$ in similar labeled edge-sets, with the set of edge-weights being a \emph{multiset}} with $\wt(e_x) = \wt(e_y)$, $e_x \in E_{T_1}$ and $e_y \in E_{T_2}$, which induce the same partition on the set of leaf-labels and edge-weights in their corresponding tree (cf. Definition \ref{def_gep}).

In \cite{DasGupta2000} this computational step is performed in $\O(n^{2})$ time which dominates the total running time of the original algorithm. In \cite{Hon2000,Hon2004} Hon et al. give an improved algorithm for computing good edge-pairs, whose running time is $\O(n\log n)$. In the following, we adopt the approach of Hon et al. and design an efficient parallel algorithm running in $\O(\log n)$ time on $\O(n\log n)$ processors. Let us first give an outline of the approach of Hon et al.

\paragraph{Partition-Labeling Problem}

In \cite{Hon2004} Hon et al. define a problem called the \emph{partition-labeling problem} on two rooted trees and present a solution running in $\O(n\log n)$ time. Then, the problem of computing good edge-pairs between two weighted phylogenies is reduced to this problem in time $\O(n\log n)$. Therefore, the time complexity of the original algorithm is improved from $\O(n^{2})$ to $\O(n\log n)$. A \emph{partition-labeling} between two rooted trees is defined as follows:

Let $R$ and $R'$ be two rooted trees with leaves labeled by the same multi-set $S$ of leaf labels. Let $A$ be any subset of $\delta(S)$, where $\delta(S)$ is the set of \emph{distinct} symbols or labels in $S$.
For each internal node $u\in V(R)$, $L_{R}(u)$ is defined as the multi-set of leaf labels in the subtree of $R$ rooted at $u$, and $L_{R}(u)|A$ to be the restriction of $L_{R}(u)$ to $A$.
Given $R$ and $R'$, let $V$ and $V'$ be the sets of internal nodes in $R$ and $R'$, respectively. A pair of mappings $\rho: V\to [1,\ell]$ and $\rho': V\to [1,\ell]$, $\ell=|V|+|V'|$, is called a \emph{partition-labeling} for $R$ and $R'$, if for all $u\in V$ and $v\in V'$, $\rho(u) = \rho'(v)$ if and only if $L_{R}(u)=L_{R'}(v)$.

The partition-labeling problem is to find a partition-labeling $(\rho,\rho')$ for $R$ and $R'$. A straightforward approach is to compute all multi-sets of $L_{R}(u)$ and $L_{R}(v)$, but this, similar to the approach of DasGupta et al., also takes $\O(n^{2})$ time. In order to reduce the time complexity, Hon et al. compute the multi-sets in an incremental manner and compare them based on earlier partial results. For this purpose, let $R_{A}$ be the contracted subtree of $R$ induced by $A$, containing only leaves with labels in $A$ and their \emph{common lowest ancestors}. Algorithm \ref{alg:Hon_etal} shows the framework of the method described by Hon et al. in \cite{Hon2004}.

\begin{algorithm}[!htb]
\KwIn{Two rooted trees $R, R'$ with leaves labeled by the same multi-set $S$.}
\KwOut{Partition-labeling $(\rho,\rho')$ for $R,R'$.}
\BlankLine
 \ForEach{$A_i\in\{A_1,A_2,\dots,A_{|\delta(S)|}\}$}{
  Compute partition-labeling for $R_{A_i}$ and $R'_{A_i}$\;
 }
 \For{$k=1$ \KwTo $\log n$}{
  Let $A_1,A_2,\dots$ be the labels considered in the last round\;
  Pair up $A_i$'s such that $A_{2j-1}=A_{2j-1}\cup A_{2j}$\;
  Delete all $A_j$'s and set $A_{2j-1}=:A_j$\;
  \ForEach{$A_j$}{
   Compute partition-labeling for $R_{A_j}$ and $R'_{A_j}$ based on the result of last round\;
  }
 }
\caption{\textsc{Partition\_Labeling}\label{alg:Hon_etal}}
\end{algorithm}

We will now show how the first foreach-phase can be efficiently computed in parallel. %, we have to prove a series of lemmas.
% First of all, we show that the induced subtree $R_A$ can be computed in time $\O(\log t)$ on $\O(t \log t)$ processors, where $t$ is the number of leaves in $R$ with labels in $A$.

\begin{Lem}\label{lem:GEPinduced}
 The induced subtree $R_A$ can be computed in time $\O(\log t)$ on $\O(t \log t)$ processors.
\end{Lem}
\begin{proof}
 Using the algorithm of Schieber and Vishkin \cite{Schieber1988}, with preprocessing in time $\O(\log n)$ on $\O(n\log n)$ processors, we can answer \emph{lowest common ancestor} queries for a pair of nodes in $R$ in time $\O(1)$. Furthermore, we use the \emph{Euler-Tour Technique} (ETT) of Tarjan and Vishkin \cite{Tarjan1984} to compute the postorder and preorder numberings of nodes in $R$ in $\O(\log n)$ time on $\O(n)$ processors.
 We construct $R_A$ as follows, given the fact that all trees under consideration are $3$-regular.
 
 Let $\ell_1,\ell_2,\dots,\ell_t$ be the sequence of leaves of $R$ with labels in $A$ and ordered from left to right by the preorder numbers $\pre(\ell_i)$.
 We perform lowest common ancestor queries for each pair $(\ell_i,\ell_{i+1})$ of the leaf-sequence and yield the set of internal nodes $w_1,w_2,\dots,w_k$ of $R_A$, i.e. $\LCA(\ell_i,\ell_{i+1})=w_j\in V(R_A)$, for all $1\leq i < t$, respectively \footnote{Note that all internal nodes of $R_A$ are found in this way, since for every internal node $w$ of $R_A$ the right-most leaf $u$ of the left subtree below $w$ is neighboring the left-most leaf $v$ of the right subtree below $w$ in terms of the preorder sequence of leaves and $\LCA(u,v)=w$.}.
 Now let $\pre(R_A)$ and $\post(R_A)$ be the (partial) preorder and postorder sequences of nodes in $R$ restricted to the leaves and internal nodes of $R_A$.
 In order to reconstruct the (contracted) topology of $R_A$, we take both sequences and generate the parental pointers $\parent(v), v\in R_A$ in parallel as follows:

 For every internal node $w$ of $R_A$ with $\pre(w)=x$ and $\post(w)=y$ we look at position $x+1$ in $\pre(R_A)$ and position $y-1$ in $\post(R_A)$ to find the right-hand child and left-hand child of $w$, respectively. This can be done in parallel for every internal node $w_i$ of $R_A$ in $\O(1)$ time on $\O(t)$ processors. This completes the construction of the induced subtree $R_A$ (in time $\O(1)$ on $\O(t)$ processors, with pre-processing in amortized $\O(\log t)$ time on $\O(t\log t)$ processors).
\end{proof}

By Lemma 2.2 in \cite{Hon2004}, we have the following fact: Let $A$ and $B$ be two disjoint subsets of $\delta(S)$ and let $u$ be an internal node in $R_{A\cup B}$. Then, $L_{R_{A\cup B}}(u)|A=\emptyset$ or $L_{R_A}(v)$ for some $v\in R_A$ and similarly, $L_{R_{A\cup B}}(u)|B=\emptyset$ or $L_{R_B}(v)$ for some $v\in R_B$.

The next lemma implies that the first foreach-phase of Algorithm \ref{alg:Hon_etal} can be completed in $\O(\log n)$ time on $\O(n\log n)$ processors.

\begin{Lem}\label{lem:GEPlabel}
 Let $a\in \delta(S)$, a partition-labeling for $R_{\{a\}}$ and $R'_{\{a\}}$ can be found in $\O(\log t)$ time on $\O(t)$ processors, where $t$ is the number of leaves in $R$ with label $a$.
\end{Lem}
\begin{proof}
 Perform a postorder numbering on $R_{\{a\}}$ in $\O(\log t)$ time on $\O(t)$ processors (cf. \cite{Tarjan1984}). Since $L_{R_{\{a\}}}(u)$ only contains multiple copies of $a$, we only need to keep track of $|L_{R_{\{a\}}}(u)|$, i.e. the number of leaves of the subtree below $u$. The number of descendant leaves for each internal vertex $u$ can be obtained from the prefix sum of the weights of edges determined in the postorder numbering algorithm. Assign this number to $u$ and apply the same procedure to $R'_{\{a\}}$.
\end{proof}

Now, we have the partition-labeling for $R_{\{i\}}$ for every distinct label $i\in\delta(S)$. In the next phase of Algorithm \ref{alg:Hon_etal} the labels are paired together and the corresponding trees $R_{\{i\}}$ and $R_{\{j\}}$ for $i,j\in\delta(S)$ are merged to form $R_{\{i,j\}}$. A partition-labeling is computed based on the partition-labeling of the first round and this is repeated for $\log |\delta(S)|$ rounds until a partition-labeling for $R_{\delta(S)}$ is produced. Let us describe the relabeling of the internal nodes of $R_{A\cup B}$ and $R'_{A\cup B}$ for two distinct subsets $A,B\subset\delta(S)$, given the corresponding partition-labelings $(\rho_A,\rho'_A)$ and $(\rho_B,\rho'_B)$ for $(R_A,R'_A)$ and $(R_B,R'_B)$, respectively.

First, we consider $R_{A\cup B}$. For each internal node $u$ in $R_{A\cup B}$, assign a $2$-tuple $(a,b)$ to $u$ such that $a$ is set to the highest integer-label of $L_{R_{A\cup B}}(u)|A$ and $b$ is set to the highest integer-label of $L_{R_{A\cup B}}(u)|B$.
If $u\in R_A$ we set $a=\rho_{A}(u)$, and if $u\in R_B$ we set $b=\rho_{B}(u)$. If $L_{R_{A\cup B}}(u)|A=\emptyset$ we set $a=0$, and if $L_{R_{A\cup B}}(u)|B=\emptyset$ we set $b=0$. It remains the case where $L_{R_{A\cup B}}(u)|A\neq\emptyset$ and $u\notin R_A$. Here, there exists a node $v$ such that $L_{R_{A\cup B}}(u)|A= L_{R_A}(v)$ and we set $a=\rho_A(v)$, which is the highest label in $L_{R_{A\cup B}}(u)|A$. The case for $L_{R_{A\cup B}}(u)|B\neq\emptyset$ and $u\notin R_B$ is treated analogously and $R'_{A\cup B}$ is treated in the same way as $R_{A\cup B}$.
After we have determined the values in $(a,b)$ for every internal node $u$ of $R_{A\cup B}$ and $R'_{A\cup B}$, the $2$-tuples are sorted and a new integer (starting from $1$) is assigned to every distinct $2$-tuple. This integer is then assigned as a label to the internal node $u$ and a partition-labeling $\rho_{A\cup B}$ ($\rho'_{A\cup B}$) for $R_{A\cup B}$ ($R'_{A\cup B}$) is obtained. Algorithm \ref{alg:PPR} shows how this can be done efficiently in parallel.

\begin{algorithm}[!htb]
\KwIn{The tree $R_{A\cup B}$ with root $r$ and $t$ leaves, parental pointers $\parent(v), v\in R_{A\cup B}$ and partition-labelings $\rho_A$ and $\rho_B$ for $R_A$ and $R_B$.}
\KwOut{Partition-labeling $\rho_{A\cup B}$ for $R_{A\cup B}$.}
\BlankLine
 \ForEach{$u\in R_{A\cup B}$ \bf parallel}{
  \If{$u\in R_A$}{
    $a(u):=\rho_A(u)$\tcc*{$a(\cdot)$ and $b(\cdot)$ are initialized with $0$}
  }
  \If{$u\in R_B$}{
    $b(u):=\rho_B(u)$\;
  }
 }
 \For{$k=1$ \KwTo $\log t$}{
  \ForEach{$u\in R_{A\cup B}$ with $\parent(u)\neq r$ \bf parallel}{
   $a(\parent(u)):=\max\{a(u), a(\parent(u))\}$\;
   $b(\parent(u)):=\max\{b(u), b(\parent(u))\}$\;
   $\parent(u):=\parent(\parent(u))$\tcc*{Pointer-Jumping}
  }
 }
 \textsc{Parallel\_Radix\_Sort}$((a(u_1),b(u_1)),\dots,(a(u_t),b(u_t)))$\tcc*{In $\O(\log t)$ time on $\O(t)$ processors}
 Let $(a(u_{i_1}), b(u_{i_1})), \dots, (a(u_{i_t}), b(u_{i_t}))$ be the sorted sequence\;
 \For{$j=1,\dots,t$ \bf parallel}{
  ${\sf left}(j):=\min\left\{ j' \left|a(u_{i_j'})=a(u_{i_j}), b(u_{i_j'})=b(u_{i_j})\right.\right\}$\;
  ${\sf right}(j):=\max\left\{ j' \left|a(u_{i_j'})=a(u_{i_j}), b(u_{i_j'})=b(u_{i_j})\right.\right\}$\;
%   ${\sf left}(j):=\min\left\{ j' \left|\begin{array}{l}a(u_{i_j'})=a(u_{i_j})\\b(u_{i_j'})=b(u_{i_j})\end{array}\right.\right\}$
%   $\rho_{A\cup B}(u_{i_j}) := j$\;
 }
 \ForEach{$u_{i_j}$ with ${\sf left}(j)=j$ \bf parallel}{
  $\rho_{A\cup B}(u_{i_j}):=k$ where $u_{i_j}$ is the $k$-th such node in the sorted order\;
 }
 \ForEach{$j=1,\dots,t$ \bf parallel}{
  $\rho_{A\cup B}(u_{i_j}):=\rho_{A\cup B}(u_{i_{{\sf left}(j)}})$\;
 }
 %Relabel each distinct $2$-tuple of the sorted sequence with an integer starting from $1$\tcc*{i.e. $\rho_{A\cup B}(u)=1, \forall u: a(u)=b(u)=0$}
\caption{\textsc{Parallel\_Partition\_Relabeling}\label{alg:PPR}}
\end{algorithm}

Let us now formulate the corresponding lemma and show that the labels assigned by Algorithm \ref{alg:PPR} form a valid partition-labeling.

\begin{Lem}\label{lem:GEPrelabel}
 Given the partition-labelings $(\rho_A,\rho'_A)$ and $(\rho_B,\rho'_B)$ for $(R_A,R'_A)$ and $(R_B,R'_B)$, we can compute partition-labelings $\rho_{A\cup B}$ and $\rho'_{A\cup B}$ for $R_{A\cup B}$ and $R_{A\cup B}$ in $\O(\log t)$ time on $\O(t)$ processors where $t$ is the number of leaves in $R_{A\cup B}$.
\end{Lem}
\begin{proof}
 In Algorithm \ref{alg:PPR}, we perform a bottom-up pointer-jumping on the internal nodes of $R_{A\cup B}$ and forward the values for both, the labels in $L_{R_{A\cup B}}|A$ and in $L_{R_{A\cup B}}|B$. At every internal node $u$, during the $\O(\log t)$ rounds, we only keep track of the highest value regarding the two label sets.
 If $u\in R_A$, then $\rho_A(u)$ is the highest label of the set $L_{R_{A\cup B}}|A$ and we correctly set $a=\rho_A(u)$. Otherwise, if $u\notin R_A$ but there exist a child $s$ of $u$ in $R_{A\cup B}$ with $L_{R_{A\cup B}}(s)|A = L_{R_A}(t)$, then $\rho_A(t)$ is the highest label of the set $L_{R_{A\cup B}}|A$ and we set $a=\rho_A(t)$. If no such child exists, then $L_{R_{A\cup B}}(u)|A=\emptyset$ an we keep the initial value $a=0$. Similarly, the values for $b$ are set by Algorithm \ref{alg:PPR} according to $L_{R_{A\cup B}}|B$.

 After the relabeling process, we have $L_{R_{A\cup B}}(u) = L_{R_{A\cup B}}(u)|A \cup L_{R_{A\cup B}}(u)|B$ and hence $L_{R_{A\cup B}}(p)=L_{R'_{A\cup B}}(q)$ if and only if the corresponding $2$-tuples assigned to $p$ and $q$ are identical. Therefore, the labels assigned to the nodes after performing \textsc{Parallel\_Radix\_Sort} (cf. \cite{Blelloch1990}) on the $2$-tuples form a valid partition labeling. \textsc{Parallel\_Radix\_Sort} and the pointer-jumping are performed in $\O(\log t)$ time on $\O(t)$ processors and the initialization of $a(u), b(u)$ is done in $\O(1)$ parallel time. After the radix sort, we use bidirectional pointer-jumping to find for each node $u_{i_j}$ the leftmost and the rightmost node in the block of $u_{i_j}$, consisting of all the nodes $u_{i_k}$ which have the same pair of labels as $u_{i_j}$. Then we first assign new labels to the leftmost nodes of all blocks. This is done by performing a pointer-jumping on these nodes, using the pointers ${\sf left}(j)$ and ${\sf right}(j)$. Finally, in $\O(1)$ time, we can also assign these labels to the remaining nodes, again using the pointers ${\sf left}(j)$. Therefore, Algorithm \ref{alg:PPR} runs in $\O(\log t)$ time on $\O(t)$ processors.
\end{proof}

By Lemma \ref{lem:GEPinduced}, \ref{lem:GEPlabel} and \ref{lem:GEPrelabel} we have that the overall complexity for the partition labeling problem is $\O(\log n)$ time on $\O(n\log n)$ processors. Next, we show how the problem of identifying good edge-pairs between phylogenies $T_1,T_2$ is reduced to the partition labeling problem between two rooted trees $R,R'$.

\paragraph{Partition-Labeling and Good Edge-Pairs}

% Recall that a pair $(e_i,e_j), e_i\in E(T_1), e_j \in E(T_2)$ is said to be a \emph{good edge-pair} if and only if $\wt_1(e_i) = \wt_2(e_j)$ and both edges induce the same partition on leaf-labels and edge-weights on $T_1$ and $T_2$. Otherwise, $e_i,e_j$ are called \emph{bad edges}.
The reduction given by Hon et al. \cite{Hon2004} starts by setting $R=T_1$ and $R'=T_2$. Then an arbitrary leaf with label $a$ is fixed, and $R$ and $R'$ are rooted at the same internal node adjacent to the leaf with label $a$. %Regarding the edge-weights of $T_1$ and $T_2$, 
Then each internal edge $e=(u,v)$ is replaced by a path $u,s,v$ and a new leaf $w$ with a unique label $\rho(w)$, adjacent to $s$. %, we attach a new internal node $s$ between $u$ and $v$ and a new leaf $w$ with a unique label $\rho(w)$.
This means, for newly added leaves $w_1,w_2$ corresponding to edges $e_1,e_2$, we have $\rho(w_1)=\rho(w_2)$ if and only if $\wt(e_1)=\wt(e_2)$. This completes the construction of $R$ and $R'$.

\begin{Lem}
 The construction of $R$ and $R'$ takes $\O(\log n)$ time on $\O(n)$ processors.
\end{Lem}
\begin{proof}
 The rooting of $R$ and $R'$ at an arbitrary node takes $\O(\log n)$ time on $\O(n)$ processors using the \emph{Euler-Tour Technique} and parallel prefix sum (cf. \cite{Tarjan1984}). In order to generate the labeled leaves that represent edge-weights, we temporarily assign for each edge $e=(u,v)$ and newly added leaf $w$ the edge-weight $\wt(e)$ to $w$. Then, we sort the sequence of leaf labels of the new leafs $w_1,w_2,\dots,w_{|E|}$ and assign a unique label $x\notin S$ such that $\rho(w_i)=\rho(w_j)$ if and only if $\wt(e_i)=\wt(e_j)$. This can also be accomplished in $\O(\log n)$ time on $\O(n)$ processors.
\end{proof}

In \cite{Hon2004} the partition-labeling is used to identify \emph{bad} edges in the trees. Here, we show how to use the labeling to compute pairs of \emph{good} edges efficiently in parallel.

Given the partition-labelings $\rho$ and $\rho'$, we first generate the sorted sequences of labels $\rho(v_1),\dots,\rho(v_{\ell})$ and $\rho'(v_1),\dots,\rho'(v_{\ell})$. Then, for every position $i$ of $\rho$ such that $v_i$ corresponds to an edge in the original tree $T_1$, we activate one processor which performs in $\O(\log n)$ time a binary search on the sequence $\rho'$ in order to check if $\rho(v_i)$ occurs as a label $\rho'(v_j)$ in the other sequence. If $v_j$ corresponds to an edge in the original tree $T_2$, then these two edges form a good edge-pair.
% Now, we have to relate the problem of finding good edge-pairs to the partition-labeling problem. In Lemma 3.2 of \cite{Hon2004} this is done in terms of \emph{bad edges}: Given a partition-labeling $(\rho,\rho')$ for $R$ and $R'$, let $e=(u,v)$ be an edge in $T_1$ (w.r.t. $T_2$) and $s$ be the unique internal node between $u$ and $v$ in $R$ (w.r.t. $R'$). Then, the edge $e$ is a \emph{bad edge} in $T_1$ (w.r.t. $T_2$) if and only if the label $\rho(s)$ (w.r.t. $\rho'(s)$) is unique in both $\rho$ and $\rho'$.

Altogether we have shown the following theorem.
\begin{Thm}
 The good edge-pairs between $T_1$ and $T_2$ can be identified in $\O(\log n)$ time on $\O(n\log n)$ processors.
\end{Thm}

\section{Summary}

We have designed a new efficient parallel approximation algorithm for the nearest-neighbor-interchange-distance (\nni) of weighted phylogenies. Based on DasGupta’s approximation algorithm \cite{DasGupta2000} our algorithm achieves an approximation ratio of $\O(\log n)$ and also constructs an associated sequence of \nni-operations. For the case that no good edge-pairs exist, our algorithm runs on a CRCW-PRAM with running time $\O(\log n)$ and $\O(n)$ processors. Furthermore, we show that the good edge-pairs between two weighted phylogenies can be identified in $\O(\log n)$ time on $\O(n\log n)$ processors.

The most challenging open problem is to settle the question if this problem is APX-hard. It would also be interesting to construct new algorithms with better approximation ratio for this problem.

\newcommand{\etalchar}[1]{$^{#1}$}

\end{document}